\documentclass[journal]{IEEEtran}
\IEEEoverridecommandlockouts
\renewcommand\IEEEkeywordsname{Index Terms}
\ifCLASSINFOpdf
\else
\fi
\usepackage{xcolor,soul,framed} 
\colorlet{shadecolor}{yellow}
\usepackage[pdftex]{graphicx}
\graphicspath{{../pdf/}{../jpeg/}}
\DeclareGraphicsExtensions{.pdf,.jpeg,.png}
\usepackage{cite}
\usepackage{amsmath,amssymb,amsfonts}
\usepackage{algpseudocode}
\PassOptionsToPackage{ruled, linesnumbered}{algorithm2e}
\usepackage{algorithm2e}
\SetKw{Break}{break out of the loop}
\usepackage{graphicx}
\usepackage{textcomp}
\usepackage{comment}
\usepackage{units}
\usepackage{url}
\usepackage{geometry}
\usepackage{bbding}
\usepackage{xcolor}
\definecolor{myblue}{rgb}{0.0, 0.5, 1.0}
\definecolor{myred}{rgb}{1.0, 0.13, 0.32}
\definecolor{mygreen}{rgb}{0.31, 0.68, 0.07}
\usepackage[pdftex]{graphicx}
\DeclareGraphicsExtensions{.pdf,.jpeg,.png}
\def\BibTeX{{\rm B\kern-.05em{\sc i\kern-.025em b}\kern-.08em
    T\kern-.1667em\lower.7ex\hbox{E}\kern-.125emX}}

\newtheorem{lemma}{Lemma}
\newtheorem{proposition}{Proposition}
\newtheorem{corollary}{Corollary}
\newtheorem{remark}{Remark}

\usepackage{pgfplots}
\usepackage{pgfplotstable}
\usepackage{tikz}
\usetikzlibrary{calc}

\usepackage{balance}
\allowdisplaybreaks

\begin{document}
   \title{\huge Hiding in Plain Sight: RIS-Aided Target Obfuscation in ISAC} 
    	\newgeometry {top=25.4mm,left=19.1mm, right= 19.1mm,bottom =19.1mm}%
\author{Ahmed Magbool,~\IEEEmembership{Graduate Student Member,~IEEE,} 
Vaibhav Kumar,~\IEEEmembership{Member,~IEEE,} 
\\ Marco Di Renzo,~\IEEEmembership{Fellow,~IEEE,} and Mark F. Flanagan,~\IEEEmembership{Senior Member,~IEEE}\vspace{-0.8cm}\thanks{Ahmed Magbool and Mark F. Flanagan are with the School of Electrical and Electronic Engineering, University College Dublin, Dublin 4, D04 V1W8, Ireland (e-mail: ahmed.magbool@ucdconnect.ie, mark.flanagan@ieee.org). \par
Vaibhav Kumar is with Engineering Division, New York University Abu Dhabi (NYUAD), Abu Dhabi 129188, UAE (e-mail: vaibhav.kumar@ieee.org).\par 
Marco Di Renzo is with Universit\'e Paris-Saclay, CNRS, CentraleSup\'elec, Laboratoire des Signaux et Syst\`emes, 3 Rue Joliot-Curie, 91192 Gif-sur-Yvette, France (marco.di-renzo@universite-paris-saclay.fr), and with King's College London, Centre for Telecommunications Research -- Department of Engineering, WC2R 2LS London, United Kingdom (marco.di\_renzo@kcl.ac.uk). \par
The work of A. Magbool and M. F. Flanagan was supported by Research Ireland under Grant Number 13/RC/2077\_P2. \par
The work of M. Di Renzo was supported in part by the European Union through the Horizon Europe project COVER under grant agreement number 101086228, the Horizon Europe project UNITE under grant agreement number 101129618, the Horizon Europe project INSTINCT under grant agreement number 101139161, and the Horizon Europe project TWIN6G under grant agreement number 101182794, as well as by the Agence Nationale de la Recherche (ANR) through the France 2030 project ANR-PEPR Networks of the Future under grant agreement NF-YACARI 22-PEFT-0005, and by the CHIST-ERA project PASSIONATE under grant agreements CHIST-ERA-22-WAI-04 and ANR-23-CHR4-0003-01.}}


\maketitle
\begin{abstract}
Integrated sensing and communication (ISAC) has been identified as a promising technology for the sixth generation (6G) of communication networks. Target privacy in ISAC is essential to ensure that only legitimate sensors can detect the target while keeping it hidden from malicious ones. In this paper, we consider a downlink reconfigurable intelligent surface (RIS)-assisted ISAC system capable of protecting a sensing region against an adversarial detector. The RIS consists of both reflecting and sensing elements, adaptively changing the element assignment based on system needs. To achieve this, we minimize the maximum sensing signal-to-interference-plus-noise-ratio (SINR) at the adversarial detector within sample points in the sensing region, by optimizing the transmit beamformer at the base station, the RIS phase shift matrix, the received beamformer at the RIS, and the division between reflecting and absorptive elements at the RIS, where the latter function as sensing elements. At the same time, the system is designed to maintain a minimum sensing SINR at each monitored location, as well as minimum communication SINR for each user. To solve this challenging optimization problem, we develop an alternating optimization approach combined with a successive convex approximation based method tailored for each subproblem. Our results show that the proposed approach achieves a~\unit[25]{dB} reduction in the maximum sensing SINR at the adversarial detector compared to scenarios without sensing area protection. Also, the optimal RIS element assignment can further improve sensing protection by~\unit[3]{dB} over RISs with fixed element configuration.
\end{abstract}

\IEEEpeerreviewmaketitle
\section{Introduction} \label{sec:intro}
The sixth generation (6G) of communication networks is envisioned to support various applications that require both precise location information and high data rates simultaneously such as autonomous driving, healthcare monitoring, precision agriculture, and the internet of things (IoT)~\cite{2021_Jiang}. This has motivated the exploration of integrated sensing and communication (ISAC), where sensing and communication are performed simultaneously using shared resources~\cite{2021_Cui,2022_Liu,2022_Zhang3}. However, ISAC introduces tradeoffs between sensing and communication performance. Studies show that reconfigurable intelligent surfaces (RISs) can mitigate these tradeoffs and enhance ISAC, especially in non-line-of-sight (NLoS) scenarios \cite{2024_Magbool2,2025_Magbool,2023_Liu2}.

Beyond sensing and communication, covertness is a critical aspect of 6G networks. Covert communication ensures that transmissions remain undetectable by wardens~\cite{2023_Chen}. Despite the importance of integrated sensing and covert communication, research in this area remains limited. The study in~\cite{2024_Ghosh} explored covert RIS-assisted ISAC, where a BS communicates with multiple legitimate users while ensuring radar detection and covertness constraints. To prevent detection, artificial noise was overlaid with the message signal, keeping the warden’s received signal statistics nearly unchanged. Similarly,~\cite{2024_Hu} proposed a beamforming framework for RIS-assisted ISAC, enhancing both covertness and sensing. Their system enables covert transmission to a legitimate receiver while simultaneously sensing a target and avoiding detection. However, to the best of the authors’ knowledge, the concept of sensing privacy for ISAC remains unexplored in the literature.

A relevant and more practical scenario involves preventing adversarial detectors from identifying the presence of targets, which is crucial for various critical applications. For instance, in military operations, maintaining stealth is essential to ensure the safety of covert missions and to protect high-value assets such as aircraft, naval vessels, and ground units~\cite{2025_Zheng}. Similarly, in surveillance and counter-surveillance, safeguarding sensitive facilities and high-profile individuals from detection is imperative to prevent security breaches and unauthorized monitoring. Moreover, autonomous vehicles must remain undetected by unauthorized systems to ensure privacy and security~\cite{2022_Hataba}.

Furthermore, in RIS-assisted mono-static ISAC systems, when the LoS path between the transmitter and target is blocked, the received echo at the base station (BS) suffers quadruple path loss, significantly weakening sensing capabilities~\cite{2019_Basar}. To mitigate this, the use of active RISs have been explored for their ability to control phase shifts and amplify signals~\cite{2023_Yu2,2023_Zhu,2023_Hao,2023_Li,2024_Zhu,2023_Chen4,2023_Salem}. However, they do not fully overcome quadruple path loss and introduce high computational complexities of the optimization algorithms due to the bi-quadratic nature of the sensing SINR function. To address these limitations, recent studies have integrated radar sensors into RISs, enhancing sensing performance across various applications and system models, including channel reconstruction~\cite{2022_Liu2,2023_Hu,2024_Chen,2023_Qian}, STAR-RIS-assisted ISAC~\cite{2023_Zhang6}, beyond-diagonal RIS-assisted ISAC~\cite{2023_Wang3}, and active RIS-assisted ISAC~\cite{2024_Magbool}. Moreover,~\cite{2022_Albanese} has investigated the use of self-configuring metasurfaces, where the sensing elements are utilized to estimate the user's location, and then the RIS phase shifts are updated based on the new location information. Despite these advancements, the adaptive reflection/sensing assignment of RIS elements in ISAC systems remains an open problem in the literature.

In this paper, we address the challenge of protecting a sensing area from an adversarial detector in an RIS-assisted ISAC system. We focus on beamforming design for a system that communicates with multiple users, senses the presence of a target within a three-dimensional (3D) sensing area, and prevents the sensing area from being monitored by an adversarial detector. The main contributions of the paper are summarized as follows.
\begin{itemize}
    \item We propose a functional adaptation of the RIS elements that enhances communication, sensing, and protection against an adversarial detector within a 3D region. The RIS consists of two dynamically assigned groups of elements: reflecting and absorptive. The reflecting elements simultaneously facilitate communication, enhance sensing, and mitigate adversarial detection, while the absorptive elements receive echo signals for target detection.
    \item We develop a joint beamforming and RIS element allocation framework to minimize the maximum sensing SINR at sample locations within the sensing region at the adversarial detector, while ensuring constraints are met on the transmit power budget, minimum sensing SINR at the RIS across all locations, minimum communication SINR for each user, unit-modulus (UM) constraints at the RIS relative elements, and the unit-norm constraint at the receive beamformer for the absorptive elements at the RIS.
    \item We propose an alternating optimization (AO)-based approach to efficiently solve the resulting optimization problem. Our framework jointly optimizes the transmit beamformer and RIS phase shifts, then jointly optimizes the receive beamformer and RIS element allocation, in an iterative manner. A successive convex approximation (SCA)-based method is employed for each sub-problem.
    \item Simulation results validate the proposed design, demonstrating a significant reduction in the adversarial detector's maximum sensing SINR within the sensing area, compared to RIS-assisted ISAC systems without sensing area protection. In addition, the results highlight the performance gains achieved through the adaptive allocation of RIS elements, as compared to fixed RIS configurations.
\end{itemize}

A related problem was previously addressed in our previous work~\cite{2024_Magbool}. However,~\cite{2024_Magbool} did not focus on the sensing protection aspect, which is a significant part of this work. Also, a single point was monitored in~\cite{2024_Magbool}, whereas this work considers a sensing area. In addition,~\cite{2024_Magbool} assumes a fixed split between RIS reflecting and sensing elements, while we optimize this division in this work. These factors result in a more complex optimization problem, and consequently, more sophisticated SCA-based methods are employed.

The rest of the paper is organized as follows. Section~\ref{sec:sys_model} presents the system model. The optimization problem formulation and our proposed solution are presented in Sections~\ref{sec:prob_form} and~\ref{sec:prop_sol}, respectively. Section~\ref{sec:siml} provides simulation results and comprehensive discussions, and Section~\ref{sec:conc} concludes the paper.

\textit{Notations:} Bold lowercase and uppercase letters denote vectors and matrices, respectively. $\Re \{ \cdot \}$, $\Im \{ \cdot \}$, $|\cdot|$ and $(\cdot)^*$ represent the real part, the imaginary part, the magnitude and the complex conjugate of a complex-valued matrix, respectively. $[\mathbf{a}]_b$ and $[\mathbf{a}]_{b:c}$ represent the $b$-th entry of $\mathbf{a}$ and the vector with entries from $b$ to $c$ of $\mathbf{a}$, respectively, while $[\mathbf{A}]_{b,c}$ is the $(b,c)$-th entry of $\mathbf{A}$. $\|\mathbf{\cdot} \|_2$ and $\|\mathbf{\cdot} \|_\textsc{F}$ represent the Euclidean vector norm and the Frobenius matrix norm, respectively. $\mathbf{(\cdot)}^\mathsf{T}$ and $\mathbf{(\cdot)}^\mathsf{H}$ denote the matrix transpose and matrix conjugate transpose, respectively. $\mathbf{I}_a$ represents the $a \times a$ identity matrix, while $\mathbf{0}_{a\times b}$ denotes a matrix of size $a \times b$ whose elements are all equal to zero. $\text{diag}(\mathbf{a})$ denotes a diagonal matrix with the elements of the vector $\mathbf{a}$ on the main diagonal. For a diagonal matrix $\mathbf{A}$, $\text{diag}(\mathbf{A})$ represents a vector whose entries consist of the diagonal elements of $\mathbf{A}$. $\text{vec}(\cdot)$ and $\text{vec}^{-1}(\cdot)$ represent the column-wise vectorization and inverse column-wise vectorization operators, respectively, while $\otimes$ is the Kronecker product. $\mathbb{C}$ indicates the set of complex numbers and $j = \sqrt{-1}$. $\mathbb{E}\{ \cdot \}$ represents the expectation operator. For the numerical set $ \mathcal{A}$, $\mathcal{A}(b)$ represents the $b$-th smallest element in the set (e.g., for $\mathcal{A} = \{ 2,5,8,9\}$, $\mathcal{A}(2) = 5$), while the notation $|A|$ represents the number of elements in the set. $\mathcal{CN}(\mathbf{A},\mathbf{B})$ represents a complex Gaussian distribution with mean $\mathbf{A}$ and covariance matrix $\mathbf{B}$. Finally, $\mathcal{O}(\cdot)$ represents the big-O notation to denote the computational complexity of an
algorithm.

\vspace{-0.5cm}

\section{System Model} \label{sec:sys_model}
As depicted in Fig.~\ref{fig:system_model}, we consider an RIS-assisted ISAC system designed to simultaneously perform three tasks: (i) transmitting data to $K$ single-antenna communication users, (ii) monitoring the existence of a single point-like target within a 3D region, and (iii) preventing an adversarial detector from monitoring the existence of a target within the same sensing region.

We assume that the direct paths between the BS and the communication users, the sensing region, and the adversarial detector are blocked due to obstacles, and the RIS facilitates signal transmission and reception. The RIS comprises $N$ elements, each with an index in the set $ \mathcal{S} = \{ 1,\dots,N\}$. Out of the $N$ RIS elements, $N_\text{r}$ are reflecting elements with indices in the set $\mathcal{R}$, and $N_\text{a} = N - N_\text{r} $ are absorptive elements with indices in the set $\mathcal{A} = \mathcal{S} \setminus \mathcal{R}$. The absorptive elements forward their received signals to a processing unit connected to the RIS for echo-signal processing and detection. Each element can operate as either a reflecting or an absorptive element, but not both simultaneously. The RIS elements can switch between reflection and absorption modes as dictated by the system design.

In this section, we present the transmitted signal model, channel model and received signal model.

\begin{figure} 
         \centering
         \includegraphics[width=1\columnwidth]{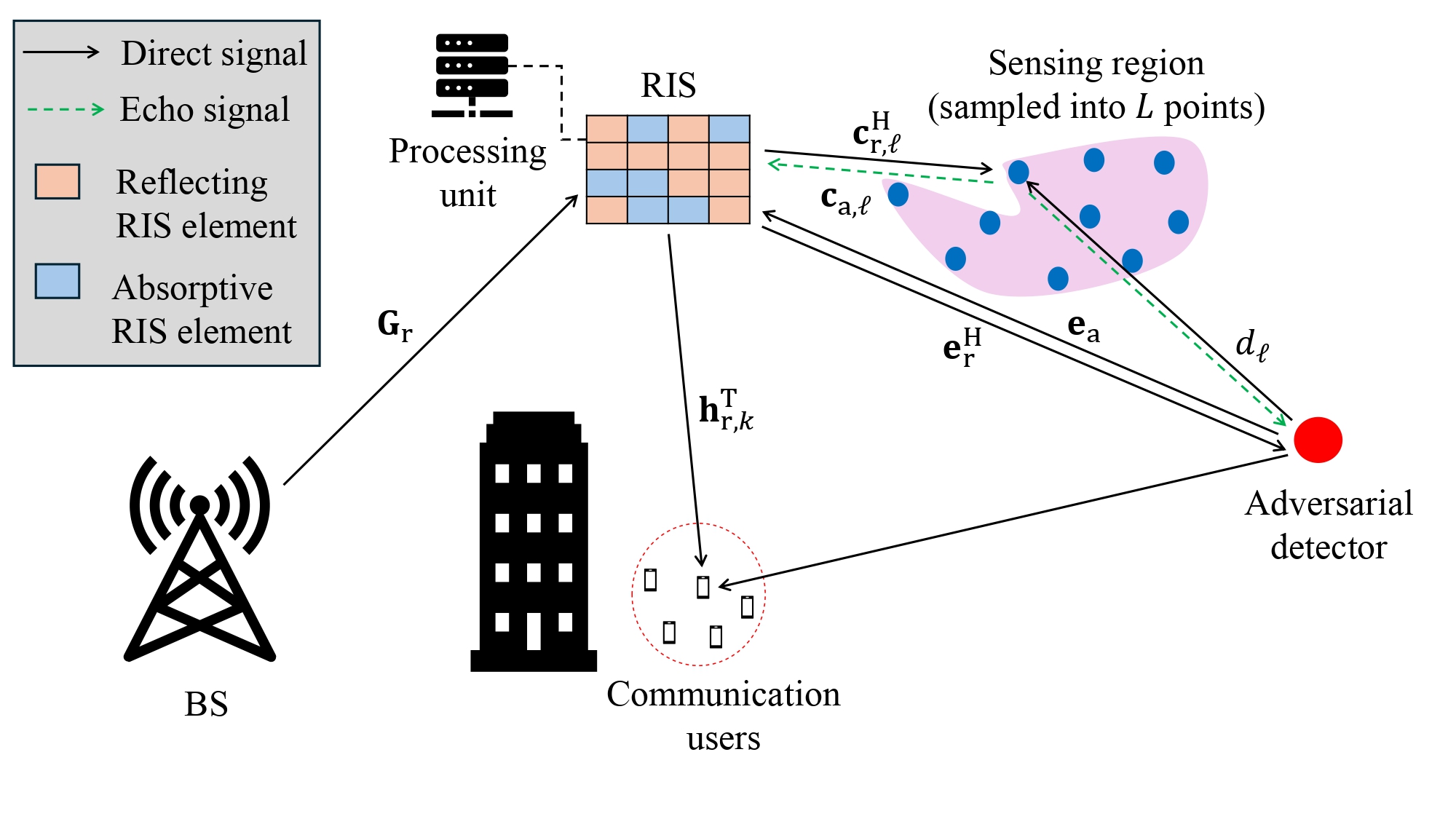}
        \caption{The proposed RIS-assisted system model consisting of a BS communicating with a number of users while detecting the presence of a target within a sensing region and simultaneously protecting the sensing region from an adversarial detector.}
        \label{fig:system_model}
\end{figure}

\subsection{Transmitted Signal Model}
We assume that the channel coherence time, during which the system channels remain unchanged, is divided into $T$ time slots. In the $t$-th time slot, the BS, equipped with $M$ antennas, transmits the following signal:
\begin{equation}
    \mathbf{x}(t) = \mathbf{W} \mathbf{s}(t), \ \  t \in \{ 1,\dots,T \},
\end{equation}
where $\mathbf{s}(t) = [s_1(t), \dots, s_K(t)]^\text{T} \in \mathbb{C}^{K \times 1}$ represents the transmitted communication symbol vector and $s_k(t)$ denotes the communication symbol intended for the $k$-th user, with $\mathbb{E} \{ \mathbf{s}(t)\mathbf{s}^\text{H}(t) \} \sim \mathcal{C}\mathcal{N}(\mathbf{0}_{K \times K}, \mathbf{I}_K)$. Additionally, $\mathbf{w}_k \in \mathbb{C}^{M \times 1}$ represents the transmit beamforming vector for the $k$-th user and $\mathbf{W} = [\mathbf{w}_1, \dots, \mathbf{w}_K] \in \mathbb{C}^{M \times K}$. The adversarial detector, on the other hand, transmits its own sensing signal denoted as $x_\text{d}(t)$, with a power level of $\varrho^2$ for each $t \in \{ 1,\dots,T\}$.

\vspace{-0.5cm}
\subsection{Channel Model}
\begin{table}
\footnotesize
\centering
\caption{Description of system channels.}
\begin{tabular} {|m{2.5cm} | m{0.7cm}| m{1.2cm} |m{2.5cm}| } 
 \hline 
 Channel & symbol & Dimensions & Description  \\  [0.5ex] 
 \hline 
 BS to RIS & $\mathbf{G}$ & $N \times M$ &  \\
 \hline
 BS to RIS reflecting elements & $\mathbf{G}_\text{r}$ &  $N_\text{r} \times M$ & Rows with indices in the set $\mathcal{R}$ of $\mathbf{G}$  \\
    \hline 
    User $k$ to RIS & $\mathbf{h}_k$ & $N \times 1$ &  \\
   \hline
   User $k$ to RIS reflecting elements & $\mathbf{h}_{\text{r},k}$ & $N_\text{r} \times 1$ & Elements with indices in the set $\mathcal{R}$ of $\mathbf{h}_k$ \\
   \hline
    Location $\ell$ to RIS & $\mathbf{c}_\ell$ & $N \times 1$ & LoS channel \\
   \hline
   Location $\ell$ to RIS reflecting elements & $\mathbf{c}_{\text{r},\ell}$ & $N_\text{r} \times 1$ & Elements with indices in the set $\mathcal{R}$ of $\mathbf{c}_\ell$ \\
   \hline
     Location $\ell$ to RIS absorptive elements & $\mathbf{c}_{\text{a},\ell}$ & $N_\text{a} \times 1$ & Elements with indices in the set $\mathcal{A}$ of $\mathbf{c}_\ell$ \\
   \hline
    Adversarial detector to Location $\ell$ & $d_\ell$ & $1 \times 1$ & LoS channel \\
   \hline
   Adversarial detector to RIS & $\mathbf{e}$ & $N \times 1$ & LoS channel \\
   \hline
     Adversarial detector to RIS reflecting elements  & $\mathbf{e}_{\text{r}}$ & $N_{\text{r}} \times 1$ & Elements with indices in the set $\mathcal{R}$ of $\mathbf{e}$ \\
   \hline
      Adversarial detector to RIS absorptive elements  & $\mathbf{e}_{\text{a}}$ & $N_{\text{a}} \times 1$ & Elements with indices in the set $\mathcal{A}$ of $\mathbf{e}$ \\
   \hline
\end{tabular}
\label{tab:channels}
\end{table}

For clarity, we summarize the system channels in Table~\ref{tab:channels}. We assume that the BS to RIS channel, $\mathbf{G}$, and the user $K$ to RIS channel, $\mathbf{h}_k$ (for each $k \in \mathcal{K} = \{1, \dots, K\}$), are available at the BS. Moreover,  we denote the interference from the adversarial detector to the $k$-th user by $f_k(t)$, which is assumed to be a Gaussian random variable (as in~\cite{2022_Jiang}) with mean $\mu_k$ and variance $\eta_k^2$, for all $k \in \mathcal{K}$ and $t \in \{1, \dots, T\}$.

To monitor the existence of targets within the sensing area, we sample this area into $L$ points~\cite{2021_Prasobh}. We assume a uniform planar array (UPA) RIS, where the horizontal and vertical distances between adjacent RIS elements are equal to $\lambda_\text{c}/2$ with $\lambda_\text{c}$ denoting the carrier wavelength. The channel between the RIS and the $\ell$-th location can be represented as~\cite{2024_Ismail,2024_Lyu,2022_Wang}
\begin{equation}
    \mathbf{c}_\ell = \alpha_\ell \mathbf{a} (\vartheta_\ell, \varphi_\ell),
\end{equation}
where $\alpha_\ell$ represents the path gain between the reference element at the RIS and the $\ell$-th location in the sensing region. The steering vector $\mathbf{a} (\vartheta_\ell, \varphi_\ell)$ is a function of the azimuth angle, $\varphi_\ell$, and the elevation angle, $\varphi_\ell$, of the $\ell$-th location relative to the reference element at the RIS. For the element in the $n_x$-th column and $n_y$-th column, the corresponding array response is
\begin{equation}
\begin{split}
     a_{n_x, n_y} (\vartheta_\ell, \varphi_\ell)   & = \exp \Big( j \pi \big[ (n_x - 1) \cos(\vartheta_\ell) \sin (\varphi_\ell) \\
    & +   (n_y - 1) \sin(\vartheta_\ell) \sin (\varphi_\ell)\big] \Big).
    \end{split}
\end{equation}

In a similar manner, the channels between the adversarial detector and the $\ell$-th location in the sensing region, $d_\ell$, for all $\ell \in  \mathcal{L} = \{1,\dots ,L \}$ as well as the channel between the RIS and the adversarial detector, denoted as $\mathbf{e}$, can be represented using the corresponding steering vectors.

\subsection{Received Signal Model}
\textit{1) Communication:} The $k$-th communication user receives the following signal:
\begin{equation}
    y_{\text{c},k}(t) = \mathbf{h}_{\text{r},k}^\text{T}  \boldsymbol{\Phi} \mathbf{G}_\text{r} \mathbf{x} (t)+  \mathbf{h}_{\text{r},k}^\text{T} \boldsymbol{\Phi}
 \mathbf{e}_\text{r} x_d (t)+ f_k(t) + n_{\text{c},k} (t),
\end{equation}
where $\boldsymbol{\Phi} \in \mathbb{C}^{N_\text{r} \times N_\text{r}}$ denotes the RIS phase shift matrix, with $|[\boldsymbol{\Phi}]_{n,n}| = 1$ for each $n \in \mathcal{S}$, and $n_{\text{c},k}(t) \sim \mathcal{C}\mathcal{N}(0, \sigma_{\text{c},k}^2)$ is the complex additive white Gaussian noise (AWGN) at User $k$.

We can obtain the communication SINR for the $k$-th user as
\begin{equation}
\begin{split}
        & \gamma_{\text{c},k} (\mathbf{W}, \boldsymbol{\Phi}, \mathcal{R}) = \\
        & \frac{| \mathbf{h}_{\text{r},k}^\text{T} \boldsymbol{\Phi} \mathbf{G}_\text{r} \mathbf{w}_k|^2}{\sum_{\underset{i\neq k}{i=1}}^K | \mathbf{h}_{\text{r},k}^\text{T} \boldsymbol{\Phi} \mathbf{G}_\text{r} \mathbf{w}_i|^2 + \varrho^2 | \mathbf{h}_{\text{r},k}^\text{T} \boldsymbol{\Phi} \mathbf{e}_\text{r}|^2  +   \bar{\sigma}_{\text{c},k}^2},
 \end{split}
\end{equation}
where $ \bar{\sigma}_{\text{c},k}^2 \triangleq \mu_k^2 + \eta_k^2 + \sigma_{\text{c},k}^2$.

\textit{2) Target Detection:} We first analyze the system's performance in detecting a target at the $\ell$-th location. The $T$ time slots of the coherence time are divided into $T/T_\text{s}$ segments, where $T/T_\text{s}$ is assumed to be an integer, during each of which the system determines the presence or absence of a target at the $\ell$-th location. Here, the target is assumed to remain either present or absent throughout each sensing segment. A large $T_\text{s}$ improves sensing performance in each detection, while a smaller $T_\text{s}$ allows for finer time resolution in detecting the target. We then analyze two cases: when $T_\text{s} = 1$ and when $T_\text{s} > 1$.

\textit{Case 1: $T_\text{s} = 1$}: we denote $\mathcal{H}_{0,\ell}$ and $\mathcal{H}_{1,\ell}$ as the null hypothesis (no target exists at the $\ell$-th location) and the alternative hypothesis (a target exists at the $\ell$-th location), respectively. Then the received signal at the RIS is expressed as
\begin{equation}
\mathbf{y}_{\text{s},\ell} (t) =
\begin{cases}
  \mathbf{y}_{\text{s}}^0 (t), & \mathcal{H}_{0,\ell}, \\
   \mathbf{y}_{\text{s},\ell}^1 (t), & \mathcal{H}_{1,\ell},
\end{cases}
\end{equation}
where
\begin{equation}
    \mathbf{y}_{\text{s}}^0 (t) = \overbrace{\mathbf{e}_\text{a} x_\text{d} (t)}^{\text{Interference from the adversarial detector}} + \mathbf{n}_{\text{s}} (t),
    \label{eq:y_0}
\end{equation}
and 
\begin{equation}
\begin{split}
    \mathbf{y}_{\text{s},\ell}^1 (t) & =  \overbrace{\beta_\ell \mathbf{c}_{\text{a},\ell} \mathbf{c}_{\text{r},\ell}^\text{H} \boldsymbol{\Phi}\mathbf{G}_\text{r} \mathbf{x} (t)}^{\substack{\text{BS-RIS-target-} \\ \text{RIS signal}}} + \overbrace{  \beta_\ell \mathbf{e}_{\text{a},\ell} d_{\ell} x_\text{d} (t)}^{\substack{\text{adversarial detector-target-} \\ \text{RIS signal}}}  \\
  & + \mathbf{e}_\text{a} x_\text{d} (t) + \mathbf{n}_{\text{s}} (t),
  \end{split}
\end{equation}
with $\beta_\ell$ denoting the target radar cross section (RCS) for the $\ell$-th location with $\mathbb{E} \{ | \beta_\ell |^2\} = \varsigma_\ell^2$ and $ \mathbf{n}_{\text{s},\ell} (t) \sim \mathcal{C} \mathcal{N} (\mathbf{0}_{N_\text{a}\times 1}, \sigma_{\text{s}}^2 \mathbf{I}_{N_\text{a}})$ represents the AWGN. We note from~\eqref{eq:y_0} that $\mathcal{H}_{0,\ell}$ is independent of the location $\ell$, so we will refer to it as $\mathcal{H}_{0}$ hereafter.

The processing unit at the RIS uses the unit-norm combiner $\mathbf{u} \in \mathbb{C}^{N_\text{a} \times 1}$ to combine the received signal. The energy of the combined signal follows the distribution:
\begin{equation}
|\mathbf{u}^\text{H} \mathbf{y}_{s,\ell} (t)|^2 \sim
\begin{cases}
 \exp (1 / \omega_0), & \mathcal{H}_0, \\
 \exp (1 / \omega_{1,\ell}), & \mathcal{H}_{1,\ell} ,
\end{cases}
\end{equation}
where $\omega_0 \triangleq \varrho^2 |\mathbf{u}^\text{H}\mathbf{e}_\text{a}|^2 + \sigma_\text{s}^2$ and $\omega_{1,\ell} \triangleq \varsigma_\ell^2 |\mathbf{u}^\text{H} \mathbf{c}_{\text{a},\ell} |^2 ||\mathbf{c}_{\text{r},\ell}^\text{H} \boldsymbol{\Phi}\mathbf{G}_\text{r}\mathbf{W}||_2^2 +\varsigma_\ell^2 \varrho^2 |d_{\ell}|^2 |\mathbf{u}^\text{H}\mathbf{c}_{\text{a},\ell} |^2 +\omega_0 $.

Thus, the processing unit makes its decision based on the following rule:
\begin{equation}
    \frac{1}{\omega_0} \text{exp} \bigg(- \frac{|\mathbf{u}^\text{H} \mathbf{y}_{s,\ell} (t)|^2}{\omega_0}\bigg)  \underset{\mathcal{C}_{1,\ell}}{\overset{\mathcal{C}_{0,\ell}}{\gtrless}}  \frac{1}{\omega_{1,\ell}} \text{exp} \bigg(- \frac{|\mathbf{u}^\text{H} \mathbf{y}_{s,\ell} (t)|^2}{\omega_{1,\ell}}\bigg),
    \label{eq:HT}
\end{equation}
where $\mathcal{C}_{0,\ell}$ and $\mathcal{C}_{1,\ell}$ represent the decisions that no target exists and that a target exists at the $\ell$-th location, respectively. We can rearrange~\eqref{eq:HT} as:
\begin{equation}
    |\mathbf{u}^\text{H} \mathbf{y}_{s,\ell} (t)|^2  \underset{\mathcal{C}_{0,\ell}}{\overset{\mathcal{C}_{1,\ell}}{\gtrless}}  \frac{\omega_{1,\ell} \omega_0}{\omega_{1,\ell} -  \omega_0} \ln \bigg( \frac{\omega_{1,\ell} }{\omega_0} \bigg) \triangleq \bar{\omega}_{\ell}.
\end{equation}

We define the sensing SINR at the RIS absorptive elements when a target is \textit{present} at the $\ell$-th location as:
\begin{equation}
\begin{split}
   &  \gamma_{\text{s},\ell} (\mathbf{W}, \boldsymbol{\Phi}, \mathbf{u}, \mathcal{R}) \\
   & = \frac{ \varsigma_\ell^2 |\mathbf{u}^\text{H} \mathbf{c}_{\text{a},\ell} |^2 ||\mathbf{c}_{\text{r},\ell}^\text{H} \boldsymbol{\Phi}\mathbf{G}_\text{r}\mathbf{W}||_2^2 +\varsigma_\ell^2 \varrho^2 |d_{\ell}|^2 |\mathbf{u}^\text{H}\mathbf{c}_{\text{a},\ell} |^2}{\varrho^2 |\mathbf{u}^\text{H}\mathbf{e}_\text{a}|^2 + \sigma_\text{s}^2}.
    \end{split}
\end{equation}
We then introduce the following proposition.

\begin{proposition} \label{prop1}
Both the false-alarm (FA) probability, $p_\ell$, and the missed-detection (MD) probability, $q_\ell$, for detecting the presence of a target at the $\ell$-th location decrease monotonically as the sensing SINR increases, following the relationships:
\begin{equation}
     p_\ell (\gamma_{\text{s},\ell}) = \big( 1 + \gamma_{\text{s},\ell} \big)^{   - ( 1 + \gamma_{\text{s},\ell}^{-1} )  },
     \label{eq:FAP_SINR}
\end{equation}
and 
\begin{equation}
     q_\ell (\gamma_{\text{s},\ell}) = 1 - \big( 1 + \gamma_{\text{s},\ell} \big)^{-\gamma_{\text{s},\ell}^{-1}},
     \label{eq:MDP_SINR}
\end{equation}
respectively.
\end{proposition}

\begin{proof}
See Appendix~\ref{AppndC}.
\end{proof}

\begin{remark}
We can observe from Proposition~\ref{prop1} that the worst-case FA and MD probabilities are:
\begin{equation}
     \lim_{ \gamma_{\text{s},\ell} \rightarrow 0} p_\ell (\gamma_{\text{s},\ell}) = e^{-1} \approx 0.37,
\end{equation}
and
\begin{equation}
     \lim_{ \gamma_{\text{s},\ell} \rightarrow 0} q_\ell  (\gamma_{\text{s},\ell}) = 1 - e^{-1} \approx 0.63.
\end{equation}
respectively. 
\end{remark}

\textit{Case 2: $T_\text{s} > 1$:} To improve detection performance, the processing unit at the RIS can make decisions by averaging the energies of the received echo signals over $T_\text{s}$ transmissions instead of relying on a single transmission. The equivalent signal in this case follows the distribution:
\begin{equation}
\frac{1}{T_\text{s}} \sum_{t=1}^{T_\text{s}} |\mathbf{u}^\text{H} \mathbf{y}_{\text{s},\ell}(t) |^2 \sim
\begin{cases} 
\text{Erlang} \Big(T_\text{s},\frac{T_\text{s}}{\omega_0} \Big), & \mathcal{H}_0, \\
\text{Erlang} \Big(T_\text{s},\frac{T_\text{s}}{\omega_{1,\ell}}\Big), & \mathcal{H}_{1,\ell}.
\end{cases}
\end{equation}

We then present the following proposition.

\begin{proposition} \label{col:1}
The FA probability, $\bar{p}_\ell$, and the MD probability, $\bar{q}_\ell$, when averaging the energies of the received echo signals over $T_\text{s}$ transmissions, can be expressed in terms of $p_\ell$ and $q_\ell$ as:
\begin{equation}
\begin{split}
        \bar{p}_\ell =  p_{\ell}^{T_\text{s}} \sum_{i=0}^{T_\text{s}-1}  \frac{\big(-T_\text{s}\ln p_{\ell} \big)^i}{i!} ,
        \end{split}
        \label{eq:FAP2}
\end{equation}
and
\begin{equation}
\begin{split}
        \bar{q}_{\ell} = 1 - (1 - q_{\ell} )^{T_\text{s}} \sum_{i=0}^{T_\text{s}-1}  \frac{\big(- T_\text{s} \ln (1 - q_{\ell})\big)^i}{i!},
        \end{split}
        \label{eq:MDP2}
\end{equation}
respectively.
\end{proposition}

\begin{proof}
The FA and MD probabilities when using $T_\text{s}$ samples can be computed using the probability density function (PDF) of the Erlang distribution as:
\begin{equation}
\begin{split}
        \bar{p}_\ell & =  \exp \Big(-T_\text{s} \frac{\bar{\omega}_{\ell}}{\omega_0}\Big) \sum_{i=0}^{T_\text{s}-1} \frac{1}{i!}  \Big(T_\text{s} \frac{\bar{\omega}_{\ell}}{\omega_0}\Big)^i,
        \end{split}
        \label{eq:FAP21}
\end{equation}
and
\begin{equation}
\begin{split}
        \bar{q}_{\ell} & =  1 - \exp \Big(-T_\text{s} \frac{\bar{\omega}_{\ell}}{\omega_{1,\ell}}\Big) \sum_{i=0}^{T_\text{s}-1} \frac{1}{i!}  \Big(T_\text{s} \frac{\bar{\omega}_{\ell}}{\omega_{1,\ell}}\Big)^i.
        \end{split}
        \label{eq:MDP21}
\end{equation}
respectively. By substituting the expression $p_\ell =\exp \Big(- \frac{\bar{\omega}_{\ell}}{\omega_0} \Big)$ (see~\eqref{eq:FAP} in Appendix~\ref{AppndC}) into~\eqref{eq:FAP21}, we obtain~\eqref{eq:FAP2}. Similarly,~\eqref{eq:MDP2} can be obtained by substituting $q_\ell = 1 - \exp  \Big(- \frac{\bar{\omega}_{\ell}}{\omega_{1,\ell}}\Big) $ (see~\eqref{eq:MDP} in Appendix~\ref{AppndC}) into~\eqref{eq:MDP21}.
\end{proof}

Therefore, the minimum sensing SINR requirement can be designed based on the desired FA and MD probabilities, given the number of samples $T_\text{s}$. This requirement will be included as a constraint later in the optimization problem.

The system applies the same concept to detect the presence of a target across the $L$ locations. As a result, a single threshold for all locations can be obtained as:
\begin{equation}
    \hat{\omega} \triangleq \min_{\ell \in \mathcal{L} } \bar{\omega}_{\ell}.
\end{equation}
This threshold is applied to make a decision on whether a target is present at \textit{any} of the $L$ locations. 
 
\textit{3) Protection against the adversarial detector:} The received echo signal at the adversarial detector from the $\ell$-th location is

\begin{equation}
y_{\text{d},\ell} (t) =
\begin{cases}
  y_{\text{d}}^0 (t), & \mathcal{H}_{0}, \\
y_{\text{d},\ell}^1 (t), & \mathcal{H}_{1,\ell},
\end{cases}
\end{equation}
where
\begin{equation}
      y_{\text{d}}^0 (t) = \overbrace{\mathbf{e}_\text{r}^\text{H} \boldsymbol{\Phi} \mathbf{G}_\text{r} \mathbf{x}(t)}^{\text{Interference from the RIS}} + n_\text{d} (t),
\end{equation}

\begin{equation}
\begin{split}
      y_{\text{d},\ell}^1 (t) & =  \overbrace{\beta_\ell |d_{\ell}|^2 x_\text{d} (t)}^{\substack{\text{Adversarial detector-target-} \\ \text{adversarial detector signal}}} + \overbrace{\beta_\ell d_{\ell} \mathbf{c}_{\text{r},\ell}^\text{H} \boldsymbol{\Phi} \mathbf{G}_\text{r} \mathbf{x} (t)}^{\substack{\text{BS-RIS-target-} \\ \text{adversarial detector signal}}}  \\
      & + \mathbf{e}_\text{r}^\text{H} \boldsymbol{\Phi} \mathbf{G}_\text{r} \mathbf{x}(t) + n_\text{d} (t),
      \end{split}
\end{equation}
 and $ n_\text{d} (t) \sim \mathcal{C} \mathcal{N}(0, \sigma_\text{d}^2) $ is the AWGN at the adversarial detector.

We then express the sensing SINR at the adversarial detector when a target is \textit{present} at the $\ell$-th location as
\begin{equation}
    \gamma_{\text{d},\ell} (\mathbf{W}, \boldsymbol{\Phi}, \mathcal{R}) = \frac{\varsigma_\ell^2 \varrho^2 |d_\ell|^4 + \varsigma_\ell^2 |d_\ell|^2 || \mathbf{c}_{\text{r},\ell}^\text{H} \boldsymbol{\Phi} \mathbf{G}_\text{r} \mathbf{W}||_2^2}{|| \mathbf{e}_{\text{r}}^\text{H} \boldsymbol{\Phi}\mathbf{G}_\text{r} \mathbf{W}||_2^2+ \sigma_\text{d}^2} .
\end{equation}
The FA and MD probabilities at the adversarial detector can be also expressed in terms of $\gamma_{\text{d},\ell}$ following Proposition~\ref{prop1}.

\section{Problem Formulation}  \label{sec:prob_form}
In this paper, our goal is to design the transmit beamformer, RIS phase shift matrix, receive beamformer, and RIS element assignments (reflecting or absorptive) to maximize the FA and MD probabilities at the adversarial detector across the $L$ locations. Based on Proposition~\ref{prop1}, both the FA and MD probabilities increase monotonically as the sensing SINR decreases. Therefore, we minimize the sensing SINR at the adversarial detector across all locations, while satisfying constraints on the system's power budget, the sensing SINRs for all locations, the communication SINRs for all users, the UM constraint at the RIS, and the unit-norm constraint on the receive beamformer. The optimization problem can be formulated as follows:
\begin{subequations}
\begin{align}
              \min_{\mathbf{W}, \boldsymbol{\Phi}, \mathbf{u}, \mathcal{R} }    \ \   &   \max_{\ell \in \mathcal{L}} \gamma_{\text{d},\ell} (\mathbf{W}, \boldsymbol{\Phi}, \mathcal{R}) \label{obj_fun} \\
          \text{s.t.} \ \ & || \mathbf{W}||_\text{F}^2 \leq P_\text{max},  \label{power_cons} \\
          &  \gamma_{\text{s},\ell} (\mathbf{W}, \boldsymbol{\Phi}, \mathbf{u}, \mathcal{R}) \geq \Gamma_{\text{s},\ell}, \ \forall \ell \in \mathcal{L},  \label{sens_cons} \\
          &  \gamma_{\text{c},k} (\mathbf{W}, \boldsymbol{\Phi}, \mathcal{R}) \geq \Gamma_{\text{c},k}, \ \ \forall k \in \mathcal{K},  \label{comm_cons} \\
           & \big| [\mathbf{\Phi}]_{n,n} \big| = 1, \  \forall n \in \mathcal{R} \label{RIS_cons}, \\
        & \| \mathbf{u} \|_2 = 1,  \label{rec_BF_cons} 
\end{align} 
\label{eq:main_opt}
\end{subequations}%
\hspace{-0.45cm} where $P_\text{max}$, $\Gamma_{\text{s},\ell}$ and $\Gamma_{\text{c},k}$ represent the transmit power budget, the sensing SINR threshold for the $\ell$-th location, and the communication SINR threshold for the $k$-th user, respectively.

Obtaining an optimal solution to the optimization problem~\eqref{eq:main_opt} is very challenging due to the non-convex nature of~\eqref{eq:main_opt}, the coupling between the optimization variables in~\eqref{obj_fun},~\eqref{sens_cons} and~\eqref{comm_cons}, the fractional forms in~\eqref{obj_fun},~\eqref{sens_cons} and~\eqref{comm_cons}, the constraints~\eqref{RIS_cons} and~\eqref{rec_BF_cons}, which are difficult to handle directly, and also due to the fact that the sets $\mathcal{R}$ and $\mathcal{A}$ affect the structure of the channels and optimization variables. In the following section, we propose an AO framework that unitizes a SCA-based method to obtain a stationary solution to~\eqref{eq:main_opt}.

\section{Proposed Solution}  \label{sec:prop_sol}
In this section, we present an SCA-based framework to obtain a stationary solution to the optimization problem in~\eqref{eq:main_opt}. The SCA method approximates the non-convex optimization problem locally with a convex one, based on the solution obtained from the previous iteration.

\subsection{Transmit Beamformer and RIS Phase Shift Matrix} \label{sec:W_PHI}
We first assume that the value of $\mathbf{u}$ and the set $\mathcal{R}$ are fixed, and we aim to find the optimal $\mathbf{W}$ and $\boldsymbol{\Phi}$. This optimization problem can be expressed as follows:
\begin{subequations}
\begin{align}
             \min_{\mathbf{W}, \boldsymbol{\Phi} }   \ \    &  \max_{\ell \in \mathcal{L}} \gamma_{\text{d},\ell} (\mathbf{W}, \boldsymbol{\Phi}) \label{obj_fun_sp1} \\
          \text{s.t.} \ \ & \eqref{power_cons}, \eqref{sens_cons}, \eqref{comm_cons},  \eqref{RIS_cons}.
\end{align} 
\label{eq:opt_sp1}
\end{subequations}%

To address~\eqref{eq:opt_sp1}, we first introduce a new variable $\boldsymbol{\psi}$, which combines $\mathbf{W}$ and $\boldsymbol{\Phi}$ as follows\footnote{For the remainder of the subsection, we will use $\boldsymbol{\psi}$ as an argument of functions in place of $\boldsymbol{\Phi}$ and $\mathbf{W}$ for brevity, as a direct one-to-one relationship exists between them. This also applies to $\boldsymbol{\psi}_0$ (in place of $\boldsymbol{\Phi}_0$ and $\mathbf{W}_0$) and $\boldsymbol{\psi}^{(\tau)}$ (in place of $\boldsymbol{\Phi}^{(\tau)}$ and $\mathbf{W}^{(\tau)}$), which will be introduced later in the paper.
}:
\begin{equation}
    \boldsymbol{\psi} \triangleq \begin{bmatrix}
        \text{vec} (\mathbf{W}) \\
        \text{diag} (\boldsymbol{\Phi}^*)
    \end{bmatrix}.
\end{equation}
We then introduce the following lemmas.

\begin{lemma}  \label{lemma:1}
For an arbitrary constant vector $\mathbf{a}$, a concave lower bound on $ \| \mathbf{a}^\textsc{H} \boldsymbol{\Phi} \mathbf{G}_\text{r} \mathbf{W}\|_2^2$ around the point $    \boldsymbol{\psi}_0 = \begin{bmatrix}
        \text{vec} (\mathbf{W}_0) \\
        \text{diag} (\boldsymbol{\Phi}_0^*)
    \end{bmatrix}$ is given by
\begin{equation}
\begin{split}
    & \| \mathbf{a}^\textsc{H} \boldsymbol{\Phi} \mathbf{G}_\text{r} \mathbf{W}\|_2^2 \geq \Omega (\mathbf{a}, \boldsymbol{\psi}, \boldsymbol{\psi}_0 )  \triangleq  - \frac{1}{2}  \|\boldsymbol{\Xi}_1 (\mathbf{a}, \boldsymbol{\psi}_0) \boldsymbol{\psi} \|_2^2 \\
    & + \Re \bigl\{ \boldsymbol{\psi}_0^\textsc{H} \boldsymbol{\Xi}_2^\textsc{H} (\mathbf{a}, \boldsymbol{\psi}_0)  \boldsymbol{\Xi}_2 (\mathbf{a}, \boldsymbol{\psi}_0)\boldsymbol{\psi} \bigl\}  - \frac{1}{2} \|\boldsymbol{\Xi}_2 (\mathbf{a}, \boldsymbol{\psi}_0) \boldsymbol{\psi}_0 \|_2^2 \\
    & - \| \boldsymbol{\Xi}_3 (\mathbf{a}, \boldsymbol{\psi}_0)\boldsymbol{\psi}_0 \|_2^2 ,
    \end{split}
    \label{eq:Lemm1}
\end{equation}
where 
\begin{subequations}
\begin{align}
 \boldsymbol{\Xi}_1 (\mathbf{a}, \boldsymbol{\psi}_0)  & \triangleq \begin{bmatrix}
    \text{diag}(\boldsymbol{\Phi_0})^\textsc{H} \mathbf{A}^\textsc{T} \mathbf{G}_\text{r}^* \mathbf{W}_0^* \otimes \mathbf{I}_M  &  - \mathbf{G}_\text{r}^\textsc{H} \mathbf{A}
\end{bmatrix}, \\
 \boldsymbol{\Xi}_2 (\mathbf{a}, \boldsymbol{\psi}_0)  & \triangleq \begin{bmatrix}
    \text{diag}(\boldsymbol{\Phi_0})^\textsc{H} \mathbf{A}^\textsc{T} \mathbf{G}_\text{r}^* \mathbf{W}_0^* \otimes \mathbf{I}_M &  \mathbf{G}_\text{r}^\textsc{H} \mathbf{A}
\end{bmatrix}, \\
 \boldsymbol{\Xi}_3 (\mathbf{a}, \boldsymbol{\psi}_0) & \triangleq \begin{bmatrix}
    \mathbf{0} & \mathbf{W}_0^\textsc{H} \mathbf{G}_\text{r}^\textsc{H} \mathbf{A}
\end{bmatrix},
\end{align}
\end{subequations}
and $\mathbf{A}  \triangleq \text{diag} (\mathbf{a})$.
\end{lemma}
\begin{proof}
See Appendix~\ref{AppndA}.
\end{proof}

\begin{corollary} \label{col:2}
 {From Lemma~\ref{lemma:1}, it follows that} for an arbitrary constant vector $\mathbf{a}$, a concave lower bound on $ | \mathbf{a}^\textsc{H} \boldsymbol{\Phi} \mathbf{G}_\text{r} \mathbf{w}_k |^2$ around the point $ \boldsymbol{\psi}_0$ is given by
\begin{equation}
\begin{split}
    & | \mathbf{a}^\textsc{H} \boldsymbol{\Phi} \mathbf{G}_\text{r} \mathbf{w}_k |^2 \geq \bar{\Omega}_k (\mathbf{a}, \boldsymbol{\psi}, \boldsymbol{\psi}_0 )  \triangleq  - \frac{1}{2}  \| \bar{\boldsymbol{\Xi}}_{1,k} (\mathbf{a}, \boldsymbol{\psi}_0) \boldsymbol{\psi} \|_2^2 \\
    & + \Re \bigl\{ \boldsymbol{\psi}_0^\textsc{H} \bar{\boldsymbol{\Xi}}_{2,k}^\textsc{H} (\mathbf{a}, \boldsymbol{\psi}_0)  \bar{\boldsymbol{\Xi}}_{2,k} (\mathbf{a}, \boldsymbol{\psi}_0)\boldsymbol{\psi} \bigl\}  - \frac{1}{2} \| \bar{\boldsymbol{\Xi}}_{2,k} (\mathbf{a}, \boldsymbol{\psi}_0) \boldsymbol{\psi}_0 \|_2^2 \\
    & - | \bar{\boldsymbol{\xi}}_k^\textsc{H} (\mathbf{a}, \boldsymbol{\psi}_0)\boldsymbol{\psi}_0 |^2 ,
    \end{split}
    \label{eq:Lemm2}
\end{equation}
where 
\begin{subequations}
\begin{align}
& \bar{\boldsymbol{\Xi}}_{1,k} (\mathbf{a}, \boldsymbol{\psi}_0) \nonumber \\
& \triangleq \begin{bmatrix}
    (\text{diag}(\boldsymbol{\Phi_0})^\textsc{H} \mathbf{A}^\textsc{T} \mathbf{G}_\text{r}^* \mathbf{w}_{k,0}^* \otimes \mathbf{I}_M) \otimes \boldsymbol{\delta}_k^\textsc{T}  &  - \mathbf{G}_\text{r}^\textsc{H} \mathbf{A}
\end{bmatrix}, \\
& \bar{\boldsymbol{\Xi}}_{2,k} (\mathbf{a}, \boldsymbol{\psi}_0) \nonumber \\
& \triangleq \begin{bmatrix}
    (\text{diag}(\boldsymbol{\Phi_0})^\textsc{H} \mathbf{A}^\textsc{T} \mathbf{G}_\text{r}^* \mathbf{w}_{k,0}^* \otimes \mathbf{I}_M) \otimes \boldsymbol{\delta}_k^\textsc{T} &  \mathbf{G}_\text{r}^\textsc{H} \mathbf{A}
\end{bmatrix}, \\
&\bar{\boldsymbol{\xi}}_k^\text{H} (\mathbf{a}, \boldsymbol{\psi}_0) \triangleq \begin{bmatrix}
    \mathbf{0}_{1 \times MK} & \mathbf{w}_{0,k}^\textsc{H} \mathbf{G}_\text{r}^\textsc{H} \mathbf{A}
\end{bmatrix},
\end{align}
\end{subequations}
 and $\boldsymbol{\delta}_k$ is a vector with the $k$-th element equal to one and all other elements equal to zero.
\end{corollary}

\begin{lemma} \label{lemma:3}
For an arbitrary constant vector $\mathbf{a}$, a convex upper bound on $ | \mathbf{a}^\textsc{H} \boldsymbol{\Phi} \mathbf{G}_\text{r} \mathbf{w}_k|^2$ around $\boldsymbol{\psi}_0$ is given by
\begin{equation}
      | \mathbf{a}^\textsc{H} \boldsymbol{\Phi} \mathbf{G}_\text{r} \mathbf{w}_k|^2 \leq \rho_k^2 + \zeta_k^2, 
\end{equation}
where $\rho_k$ and $\zeta_k$ are slack variables satisfying the following constraints:
\begin{subequations}
\begin{align}
     \rho_k & \geq \Lambda_k (\pm \mathbf{a},\boldsymbol{\psi},\boldsymbol{\psi}_0), \label{eq:in11} \\
     \rho_k & \geq \Lambda_k ( \pm j\mathbf{a},\boldsymbol{\psi},\boldsymbol{\psi}_0), \label{eq:in21}
\end{align}
\end{subequations}
where $\Lambda_k (\mathbf{a},\boldsymbol{\psi},\boldsymbol{\psi}_0) \triangleq \frac{1}{4} \|  \boldsymbol{\Pi}_{k} (\mathbf{a}) \boldsymbol{\psi}\|_2^2 - \frac{1}{2} \Re \{ \boldsymbol{\psi}_0^\textsc{H} \boldsymbol{\Pi}_{k}^\textsc{H} (-\mathbf{a}) \boldsymbol{\Pi}_{k} (-\mathbf{a}) \boldsymbol{\psi} \} + \frac{1}{4} \|  \boldsymbol{\Pi}_{k} (- \mathbf{a}) \boldsymbol{\psi}_0\|_2^2$ and $\boldsymbol{\Pi}_{k} ( \mathbf{a}) \triangleq \begin{bmatrix}
    \mathbf{A}^*\mathbf{G}_\text{r} \otimes \boldsymbol{\delta}_k^\textsc{T} &  \mathbf{I}_{N_\text{r}}
\end{bmatrix}$.

\end{lemma}
\begin{proof}
See Appendix~\ref{AppndB}.
\end{proof}

Now, we start describing the SCA method by addressing the objective function in~\eqref{obj_fun_sp1}. We first observe that the term $\varsigma_\ell^2 |d_\ell|^2 || \mathbf{c}_{\text{r},\ell}^\text{H} \boldsymbol{\Phi} \mathbf{G}_\text{r} \mathbf{W}||_2^2$ in~\eqref{obj_fun_sp1} has a negligible impact compared to $|| \mathbf{e}_{\text{r}}^\text{H} \boldsymbol{\Phi}\mathbf{G}_\text{r} \mathbf{W}||_2^2$ due to its significant path loss. Consequently, maximizing the direct interference power from the RIS to the adversarial detector is sufficient to minimize the sensing SINR across all locations (i.e., maximizing the term $|| \mathbf{e}_{\text{r}}^\text{H} \boldsymbol{\Phi}\mathbf{G}_\text{r} \mathbf{W}||_2^2$).

We apply Lemma~\ref{lemma:1} to derive a surrogate function for the term $|| \mathbf{e}_{\text{r}}^\text{H} \boldsymbol{\Phi}\mathbf{G}_\text{r} \mathbf{W}||_2^2$ as follows:
\begin{equation}
    || \mathbf{e}_{\text{r}}^\text{H} \boldsymbol{\Phi}\mathbf{G}_\text{r} \mathbf{W}||_2^2  \geq \Omega (\mathbf{e}_\text{r}, \boldsymbol{\psi}, \boldsymbol{\psi}^{(\tau-1)} ),
    \label{eq:trans_obj1}
\end{equation}
where $    \boldsymbol{\psi}^{(\tau)} = \begin{bmatrix}
        \text{vec} (\mathbf{W}^{(\tau)}) \\
        \text{diag} ({\boldsymbol{\Phi}^{(\tau)}}^*)
    \end{bmatrix}$ is the value of $\boldsymbol{\psi}$ at the $\tau$-th iteration of the SCA method.

Next, we focus on the constraint set~\eqref{sens_cons}, which we first rewrite for the $\ell$-th location as follows:
\begin{equation}
    ||\mathbf{c}_{\text{r},\ell}^\text{H} \boldsymbol{\Phi}\mathbf{G}_\text{r}\mathbf{W}||_2^2 \geq \bar{\Gamma}_{\text{s},\ell},
    \label{eq:sensing_SINR_inq}
\end{equation}
where $\bar{\Gamma}_{\text{s},\ell} \triangleq \Gamma_{\text{s},\ell} (\varrho^2 |\mathbf{u}^\text{H}\mathbf{e}_\text{a}|^2 + \sigma_\text{s}^2) / \varsigma_\ell^2 |\mathbf{u}^\text{H} \mathbf{c}_{\text{a},\ell} |^2 - \varrho^2 |d_{\ell}|^2$. We then apply Lemma~\ref{lemma:1} to derive a concave surrogate function for the left-hand side of~\eqref{eq:sensing_SINR_inq}. This transforms the constraint in~\eqref{sens_cons} to the following form:
\begin{equation}
    \Omega (\mathbf{c}_{\text{r},\ell}, \boldsymbol{\psi}, \boldsymbol{\psi}^{(\tau-1)} )\geq \bar{\Gamma}_\ell, \ \ \forall \ell \in \mathcal{L},
    \label{sen_cons_modJ}
\end{equation}
which is a convex constraint set in $\boldsymbol{\psi}$.

Next, we address the constraint set~\eqref{comm_cons}, which we rewrite for the $k$-th user as 
\begin{equation}
\begin{split}
         & | \mathbf{h}_{\text{r},k}^\text{T} \boldsymbol{\Phi} \mathbf{G}_\text{r} \mathbf{w}_k|^2 \\
         & \geq \Gamma_{\text{c},k} \bigg( \sum_{\underset{i\neq k}{i=1}}^K | \mathbf{h}_{\text{r},k}^\text{T} \boldsymbol{\Phi} \mathbf{G}_\text{r} \mathbf{w}_i|^2 + \varrho^2 | \mathbf{h}_{\text{r},k}^\text{T} \boldsymbol{\Phi} \mathbf{e}_\text{r}|^2  +   \bar{\sigma}_{\text{c},k}^2 \bigg).
 \end{split}
 \label{eq:comm_SINR_W}
\end{equation}
The left-hand side of~\eqref{eq:comm_SINR_W} can be addressed using Corollary~\ref{col:2} as follows:
\begin{equation}
    | \mathbf{h}_{\text{r},k}^\text{T} \boldsymbol{\Phi} \mathbf{G}_\text{r} \mathbf{w}_k|^2 \geq  \bar{\Omega} (\mathbf{h}_{\text{r},k}^*,\boldsymbol{\psi}, \boldsymbol{\psi}^{(\tau-1)}).
\end{equation}
Furthermore, we apply Lemma~\ref{lemma:3} to derive a surrogate functions for the $i$-th term in the summation on the right-hand side of~\eqref{eq:comm_SINR_W} by introducing the optimization variables ${ \delta_{k,i} }$ and ${ \chi_{k,i} }$ that satisfies the following conditions:
\begin{subequations}
\begin{align}
     \delta_{k,i} & \geq \Lambda_i (\pm \mathbf{h}_{\text{r},k}^*, \boldsymbol{\psi}, \boldsymbol{\psi}^{(\tau-1)}), \label{eq:in1} \\
     \chi_{k,i} & \geq \Lambda_i (\pm j  \mathbf{h}_{\text{r},k}^*, \boldsymbol{\psi}, \boldsymbol{\psi}^{(\tau-1)}), \label{eq:in2} 
\end{align}
\end{subequations}
Thus, we can approximate~\eqref{comm_cons} by a convex constraint as
\begin{equation}
\begin{split}
   & \bar{\Omega} (\mathbf{h}_{\text{r},k}^*,\boldsymbol{\psi}, \boldsymbol{\psi}^{(\tau-1)})  \geq \Gamma_{\text{c},k} \big( \sum_{\underset{i\neq k}{i=1}}^K \delta_{k,i}^2 + \xi_{k,i}^2 \\
    & + \varrho^2 \big| \begin{bmatrix}
        \mathbf{0}_{1 \times MK} & \mathbf{h}_{\text{r},k}^\text{T} \mathbf{E}_\text{r}^*
    \end{bmatrix} \boldsymbol{\psi} \big|^2 +   \bar{\sigma}_{\text{c},k}^2 \big), \ \ \forall k \in \mathcal{K},
    \end{split}
    \label{eq:comm_SINR_m_sp1}
\end{equation}
where $\mathbf{E}_\text{r} \triangleq \text{diag}(\mathbf{e}_\text{r})$.

Finally, the UM constraint at the RIS~\eqref{RIS_cons} can be relaxed as follows:
\begin{equation}
    |[\boldsymbol{\Phi}]_{n,n} | \leq 1, \ \forall n \in \mathcal{R}.
    \label{eq:relaxed_RIS1}
\end{equation}
To ensure that the relaxed constraint~\eqref{eq:relaxed_RIS1} is binding at convergence, we add a penalty term with a small regularization parameter to the objective function. This penalty term is expressed as:
\begin{equation}
    \mathcal{H} (\boldsymbol{\Phi}) = || \boldsymbol{\Phi}  ||_\text{F}^2.
    \label{eq:theta_Phi}
\end{equation}
We then derive a linear surrogate function for~\eqref{eq:theta_Phi} using the Taylor series expansion as follows\footnote{From our simulations, we observed that the relaxed constraint in~\eqref{eq:relaxed_RIS1} leads to the original constraint in~\eqref{RIS_cons}, even without including the penalty term in the objective function. However, we include it here for completeness.}:
\begin{equation}
\begin{split}
    & \mathcal{H} (\boldsymbol{\Phi}) \geq   2 \overbrace{\Re \big\{ \text{Tr} \big( (\boldsymbol{\Phi}^{(\tau-1)})^* \boldsymbol{\Phi} \big) \big\}}^{\bar{\mathcal{H}} (\boldsymbol{\Phi},\boldsymbol{\Phi}^{(\tau-1)})} - || \boldsymbol{\Phi}^{(\tau-1)}  ||_\text{F}^2.
    \end{split}
\end{equation}
We then re-express $\bar{\mathcal{H}} (\boldsymbol{\Phi},\boldsymbol{\Phi}^{(\tau-1)})$ in terms of $\boldsymbol{\psi}$ and $\boldsymbol{\psi}^{(\tau-1)}$ as
\begin{equation}
\begin{split}
        &\bar{\mathcal{H}} (\boldsymbol{\Phi},\boldsymbol{\Phi}^{(\tau-1)})  = \bar{\mathcal{H}} (\boldsymbol{\psi},\boldsymbol{\psi}^{(\tau-1)}) \\
        & = \Re \{ \big( [ \boldsymbol{\psi}]_{MK+1:MK+N_\text{r}}^{(\tau-1)} \big)^\text{H} [\boldsymbol{\psi}]_{MK+1:MK+N_\text{r}} \big\} \}.
        \end{split}
\end{equation}

Therefore, the optimization problem of the transmit beamformer and RIS phase shift matrix at the $\tau$-th iteration of the SCA algorithm is formulated as follows:
\begin{subequations}
\begin{align}
             & \max_{ \boldsymbol{\psi}, \{\delta_{k,i}  \} , \{\chi_{k,i}  \} }   \      \Omega (\mathbf{e}_\text{r}, \boldsymbol{\psi}, \boldsymbol{\psi}^{(\tau-1)} ) + \varepsilon  \bar{\mathcal{H}} (\boldsymbol{\psi},\boldsymbol{\psi}^{(\tau-1)}) \\
           \text{s.t.} \  & || [\boldsymbol{\psi}]_{1:MK} ||_2^2 \leq P_\text{max},  \\
          &  \eqref{sen_cons_modJ},  \eqref{eq:comm_SINR_m_sp1},  \\
          & \eqref{eq:in1}, \eqref{eq:in2}, \ \forall i \in \mathcal{K} \ \backslash \{ k \},  \ \forall k \in \mathcal{K},  \\
           & |[\boldsymbol{\psi}]_{n} | \leq 1, \ \forall n \in \{MK +1,\dots, MK +N_\text{r}\}, 
\end{align} 
\label{eq:opt_JO}
\end{subequations}%
\hspace{-0.2cm} where $\varepsilon$ is a small regularization parameter. This optimization problem is a second-order cone program (SOCP) and can be solved efficiently using the CVX toolbox. We summarize the joint optimization of the transmit beamformer and RIS phase shift matrix in~\textbf{Algorithm~\ref{algo1}}.

\begin{algorithm}[t]
\caption{Transmit beamformer and RIS phase shift matrix optimization} \label{algo1}

\KwIn{  $\mathbf{u}$, $\mathcal{R}$, $\boldsymbol{\psi}^{(0)} = \begin{bmatrix}
        \text{vec} (\mathbf{W}^{(0)}) \\
        \text{diag} \big(\boldsymbol{\Phi}^{(0)} \big)^*
    \end{bmatrix}$, $\varepsilon \geq 0$ , $\epsilon \geq 0$, $\tau = 0$ }

\Repeat{ $ \| \boldsymbol{\psi}^{(\tau)} - \boldsymbol{\psi}^{(\tau-1)}  \|_2^2 \leq \epsilon $}{
    
$ \tau \leftarrow \tau + 1$\;

Compute $\boldsymbol{\psi}^{(\tau)}$ by solving~\eqref{eq:opt_JO}\; 

}

\KwOut{ $\mathbf{W} = \text{vec}^{-1} \big([\boldsymbol{\psi}^{(\tau)}]_{1:MK}\big) $, $\boldsymbol{\Phi} = \text{diag} \big([\boldsymbol{\psi}^{(\tau)}]_{MK+1:MK+N_\text{r}}\big)^*$}
\end{algorithm}

\subsection{Receive Beamforming and RIS Element Assignment} \label{sec:u_R_A}
Now, assuming the values of $\mathbf{W}$ and $\boldsymbol{\Phi}$ are fixed, we aim to optimize for $\mathbf{u}$ and $\mathcal{R}$. The corresponding sub-optimization problem for these variables is:
\begin{subequations}
\begin{align}
             \min_{ \mathbf{u}, \mathcal{R}}   \ \    &  \Tilde{\alpha} \label{obj_fun_sp2} \\
          \text{s.t.} \ \ 
          &  \eqref{sens_cons}, \eqref{rec_BF_cons},
\end{align} 
\label{eq:SP2}
\end{subequations}%
\hspace{-0.15cm}where $\Tilde{\alpha}$ is a constant. Notice that any feasible point is an optimal solution to~\eqref{eq:SP2}, as the objective function is constant. However, it can be observed from~\eqref{eq:main_opt} that increasing the number of absorptive elements may only improve the sensing SINR at the RIS. On the other hand, increasing the number of reflecting elements can enhance the sensing SINRs at the RIS, the communication SINRs, and the objective function. Therefore, it is desirable to maximize the number of reflecting elements without violating the constraint set~\eqref{sens_cons}.

 {
Based on this argument, we need to obtain $\mathbf{u}$ and $\mathcal{R}$ that satisfy the constraint set~\eqref{sens_cons} while maximizing the number of reflecting elements (i.e., maximizing $|\mathcal{R}|$), which is challenging to tackle directly. However, we observe that maximizing $|\mathcal{R}|$ is equivalent to minimizing $|\mathcal{A}|$. To approximate this optimization problem, we minimize the number of nonzero elements in the receive combining vector $\mathbf{u}$ while ensuring compliance with the constraints~\eqref{sens_cons} and~\eqref{rec_BF_cons}. This leads to the following optimization problem:
\begin{subequations}
\begin{align}
             \min_{ \mathbf{u} }   \ \    &  \| \mathbf{u} \|_0 \label{obj_fun_sp3} \\
          \text{s.t.} \ \ 
          &  \eqref{sens_cons}, \eqref{rec_BF_cons}.
\end{align} 
\label{eq:SP3}
\end{subequations}%
\hspace{-0.15cm}}This optimization problem seeks the smallest-dimension receive beamformer that satisfies~\eqref{sens_cons} and~\eqref{rec_BF_cons}, optimizing the nonzero elements of $\mathbf{u}$ to maximize the left-hand side of~\eqref{sens_cons} while adhering to~\eqref{rec_BF_cons}. However, the optimization problem~\eqref{eq:SP3} is non-convex. To address this, we first approximate the objective function~\eqref{obj_fun_sp3} by replacing the $l_0$ norm with the $l_1$ norm, which also promotes sparsity~\cite{2013_Ramirez}.

Next, we re-express the constraint set~\eqref{sens_cons} as
\begin{equation}
     |\mathbf{u}^\text{H} \mathbf{c}_{\text{a},\ell} |^2 \geq \frac{\Gamma_{\text{s},\ell} (\varrho^2 |\mathbf{u}^\text{H}\mathbf{e}_\text{a}|^2 + \sigma_\text{s}^2)}{\varsigma_\ell^2 ||\mathbf{c}_{\text{r},\ell}^\text{H} \boldsymbol{\Phi}\mathbf{G}_\text{r}\mathbf{W}||_2^2 +\varsigma_\ell^2 \varrho^2 |d_{\ell}|^2}.
     \label{eq:tr_se_co}
\end{equation}
While the right-hand side of~\eqref{eq:tr_se_co} is already convex in $\mathbf{u}$, we can derive a surrogate function to lower-bound the left-hand side of~\eqref{eq:tr_se_co} using its first-order Taylor series expansion. Thus,~\eqref{eq:tr_se_co} can be approximated as
\begin{equation}
    \begin{split}
        & 2 \Re \{ (\mathbf{u}^{(\tau-1)})^\text{H} \mathbf{c}_{\text{a},\ell} \mathbf{c}_{\text{a},\ell}^\text{H}  \mathbf{u} \} - | \mathbf{c}_{\text{a},\ell}^\text{H}  \mathbf{u}^{(\tau-1)}|^2 \\
        & \geq \frac{\Gamma_{\text{s},\ell} (\varrho^2 |\mathbf{u}^\text{H}\mathbf{e}_\text{a}|^2 + \sigma_\text{s}^2)}{\varsigma_\ell^2 ||\mathbf{c}_{\text{r},\ell}^\text{H} \boldsymbol{\Phi}\mathbf{G}_\text{r}\mathbf{W}||_2^2 +\varsigma_\ell^2 \varrho^2 |d_{\ell}|^2},
    \end{split}
    \label{eq:nc_1}
\end{equation}
which is convex in $\mathbf{u}$.

Finally, we tackle the constraint~\eqref{rec_BF_cons} by rewriting it as two separate constraints, as follows:
\begin{subequations}
\begin{align}
     \| \mathbf{u} \|_2^2 \leq 1, \label{eq:EN_1} \\
     \| \mathbf{u} \|_2^2 \geq 1. \label{eq:EN_2}
\end{align}
\end{subequations}
While~\eqref{eq:EN_1} is already convex, we can approximate~\eqref{eq:EN_2} by a linear one as follows:
\begin{equation}
    2 \Re \{ (\mathbf{u}^{(\tau-1)})^\text{H} \mathbf{u} \} - \| \mathbf{u}^{(\tau-1)} \|_2^2 \geq 1.
    \label{eq:EN_3}
\end{equation}

Therefore, in the $\tau$-th iteration of the SCA method, the optimization problem~\eqref{eq:SP3} can be approximated as follows:
\begin{subequations}
\begin{align}
             \min_{ \mathbf{u} }   \ \    &  \| \mathbf{u} \|_1  \\
          \text{s.t.} \ \ 
          &  \eqref{eq:nc_1}, \ \ \forall \ell \in \mathcal{L}, \\
          & \eqref{eq:EN_1}, \eqref{eq:EN_3},
\end{align} 
\label{eq:SP4}
\end{subequations}%
\hspace{-0.15cm}which is a convex optimization problem in $\mathbf{u}$ and can be solved using standard convex optimization techniques such as the interior-point method, which can be implemented using the CVX toolbox. 

 {
At convergence, the optimization problem~\eqref{eq:SP4} results in some entries of $\mathbf{u}$ being very small or equal to zero. We construct the new set $\mathcal{D}$ by including the elements of $\mathcal{A}$ where the magnitudes of the corresponding elements of $\mathbf{u}$ are smaller than a given threshold. The sets of reflecting and absorptive elements are then updated as $\mathcal{R} \leftarrow \mathcal{R} \cup \mathcal{D}$ and $\mathcal{A} \leftarrow \mathcal{A} \backslash \mathcal{D}$, respectively. After that, we also update the channel $\mathbf{G}_\text{r}$ by including the relevant rows from $\mathbf{G}$, as specified in Table~\ref{tab:channels}. Similarly, the channels $\mathbf{h}_{\text{r},k}$, $\mathbf{c}_{\text{r},\ell}$, and $\mathbf{e}_{\text{r}}$ are updated by including the relevant elements of $\mathbf{h}_{k}$, $\mathbf{c}_{\ell}$, and $\mathbf{e}$, respectively. Conversely, the channels $\mathbf{c}_{\text{a},\ell}$ and $\mathbf{e}_{\text{a}}$ are updated by omitting the relevant elements from $\mathbf{c}_{\ell}$ and $\mathbf{e}$, respectively. Furthermore, we update $\mathbf{u}$ and $\boldsymbol{\Phi}$ by omitting and adding the relevant elements, respectively\footnote{To ensure that the newly added RIS elements do not contribute negatively to the objective function by introducing detrimental interference, their magnitudes can be initially set to zero. This issue will be resolved when solving for $\mathbf{W}$ and $\boldsymbol{\Phi}$ in the next iteration.}. The proposed method for optimizing the receive beamforming and RIS element assignment is summarized in \textbf{Algorithm~\ref{algo2}}.}

\begin{algorithm}[t]
\caption{Receive beamforming and ris RIS element assignment optimization} \label{algo2}

\KwIn{ $\boldsymbol{\Phi}$, $\mathbf{W}$, $\mathbf{u}^{(0)}$, $\mathcal{R}^{(0)}$, $\tau = 0$, $\epsilon> 0$, $\varepsilon> 0$ }

\Repeat{ $ || \mathbf{u}^{(\tau-1)} - \mathbf{u}^{(\tau)} ||_2^2 \leq \epsilon$}{
    
$ \tau \leftarrow \tau + 1$\;

Compute $\mathbf{u}^{(\tau)}$ by solving~\eqref{eq:SP4}\; 

}

$\mathcal{D} = \emptyset$\;

\For{ $i \in  1,\dots, N_\mathsf{a}$}{
\If{ $|[ \mathbf{u}^{(\tau)} ]_i | \leq \varepsilon$}{

$\mathcal{D} \leftarrow \mathcal{D} \cup \mathcal{A}(i) $\;

}

}

$\mathcal{R}^{(\tau)} \leftarrow \mathcal{R}^{(\tau-1)} \cup \mathcal{D}$ \;

$\mathcal{A} \leftarrow \mathcal{A} \backslash \mathcal{D}$ \;

Remove the elements with indices in $\mathcal{D}$ from $\mathbf{u}^{(\tau)}$\;

Update the channels $\mathbf{G}_\text{r}$, $\mathbf{h}_{\text{r},k}$, $\mathbf{c}_{\text{r},\ell}$, $\mathbf{c}_{\text{a},\ell}$, $\mathbf{e}_{\text{r}}$ and $\mathbf{e}_{\text{a}}$ as specified in Table~\ref{tab:channels}\;

 Expand $\mathbf{\Phi}$ to include the elements in the set $\mathcal{D}$ with zero magnitudes\;

\KwOut{ $\mathbf{u} = \mathbf{u}^{(\tau)}$, $\mathcal{R} = \mathcal{R}^{(\tau)}$, $\mathcal{A} $, $\mathbf{G}_\text{r}$, $\mathbf{h}_{\text{r},k}$, $\mathbf{c}_{\text{r},\ell}$, $\mathbf{c}_{\text{a},\ell}$, $\mathbf{e}_{\text{r}}$, $\mathbf{e}_{\text{a}}$, $\mathbf{\Phi}$ }
\end{algorithm}

\subsection{Initialization} \label{sec:init}
The SCA method requires an initial feasible point, which can be non-trivial to find. In this section, we present a two-step approach to determine an initial feasible point for problem~\eqref{eq:main_opt}.

In the first step, we choose random $\mathbf{W}$, $\mathbf{\Phi}$, and $\mathbf{u}$ that satisfy the constraints~\eqref{power_cons},~\eqref{RIS_cons}, and~\eqref{rec_BF_cons}, respectively, which is a straightforward task. Similarly, we start with an initial $\mathcal{R}$. Next, we focus on refining this initialization to satisfy the constraints~\eqref{sens_cons} and~\eqref{comm_cons}\footnote{A judicious choice of initial point would make the second step more efficient. For instance, while $\boldsymbol{\Phi}$ can be chosen as the identity matrix, $\mathbf{W}$ can be chosen based on the zero-forcing (ZF) criterion to suppress inter-user interference. Moreover, $\mathbf{u}$ can be chosen to mitigate direct interference from the adversarial detector to the absorptive elements at the RIS. More importantly, the dimensions of the set $\mathcal{R}$ should be reasonable.}.

In the second step, we introduce the constraint set
\begin{equation}
\begin{split}
   & \bar{\lambda}_{\text{s},\ell} + \varsigma_\ell^2 |\mathbf{u}^\text{H} \mathbf{c}_{\text{a},\ell} |^2 ||\mathbf{c}_{\text{r},\ell}^\text{H} \boldsymbol{\Phi}\mathbf{G}_\text{r}\mathbf{W}||_2^2 +\varsigma_\ell^2 \varrho^2 |d_{\ell}|^2 |\mathbf{u}^\text{H}\mathbf{c}_{\text{a},\ell} |^2 \\
   & \geq \Gamma_{\text{s},\ell} (\varrho^2 |\mathbf{u}^\text{H}\mathbf{e}_\text{a}|^2 + \sigma_\text{s}^2),
    \end{split}
    \label{eq:in_con1}
\end{equation}
where $ \bar{\lambda}_{\text{s},\ell} \geq 0$ for all $\ell \in \mathcal{L}$. It can be noted that the constraint set~\eqref{eq:in_con1} is equivalent to~\eqref{sens_cons} when $ \bar{\lambda}_{\text{s},\ell} \rightarrow  0$ for all $\ell \in \mathcal{L}$. 

Similarly, the constraint set
\begin{equation}
\begin{split}
        & \bar{\lambda}_{\text{c},k} +  | \mathbf{h}_{\text{r},k}^\text{T} \boldsymbol{\Phi} \mathbf{G}_\text{r} \mathbf{w}_k|^2\\
        & \geq  \Gamma_{\text{c},k} \bigg( \sum_{\underset{i\neq k}{i=1}}^K | \mathbf{h}_{\text{r},k}^\text{T} \boldsymbol{\Phi} \mathbf{G}_\text{r} \mathbf{w}_i|^2 + \varrho^2 | \mathbf{h}_{\text{r},k}^\text{T} \boldsymbol{\Phi} \mathbf{e}_\text{r}|^2  +   \bar{\sigma}_{\text{c},k}^2 \bigg),
 \end{split}
  \label{eq:in_con2}
\end{equation}
with $ \bar{\lambda}_{\text{c},k} \geq 0$ for all $k \in \mathcal{K}$ is equivalent to~\eqref{comm_cons} when $ \bar{\lambda}_{\text{c},k} \rightarrow  0$ for all $k \in \mathcal{K}$. 

Thus, to refine the initial point from the first step, we solve the following optimization problem:
\begin{subequations}
\begin{align}
             & \min_{ \{ \bar{\lambda}_{\text{s},\ell} \}, \{ \bar{\lambda}_{\text{c},k},  \}, \mathbf{W}, \boldsymbol{\Phi}, \mathbf{u}, \mathcal{R}}    \ \      \sum_{\ell = 1}^L \bar{\lambda}_{\text{s},\ell} + \sum_{k = 1}^K \bar{\lambda}_{\text{c},k} \\
          \text{s.t.} \ \ &  \eqref{eq:in_con1},  \bar{\lambda}_{\text{s},\ell} \geq 0 , \ \forall \ell \in \mathcal{L},  \\
          &  \eqref{eq:in_con2}, \bar{\lambda}_{\text{c},k} \geq 0 \ \ \forall k \in \mathcal{K},  \\
           & \eqref{power_cons}, \eqref{RIS_cons}, \eqref{rec_BF_cons}.
\end{align} 
\label{eq:init}
\end{subequations}%

The optimization problem~\eqref{eq:init} can be solved using the SCA method by alternating between jointly optimizing $\{ \bar{\lambda}_{\text{s},\ell} \}, \{ \bar{\lambda}_{\text{c},k},  \}, \mathbf{W}$ and $\boldsymbol{\Phi}$ and solving for $\mathbf{u}$ and $\mathcal{R}$ using methods similar to those presented in Sections~\ref{sec:W_PHI} and~\ref{sec:u_R_A}. Obtaining values close to zero for $\{ \bar{\lambda}_{\text{s},\ell} \}$ and $ \{ \bar{\lambda}_{\text{c},k} \}$ produces a feasible point for~\eqref{eq:main_opt}. \textbf{Algorithm~\ref{algo_main}} summarizes the overall proposed approach to optimize~\eqref{eq:main_opt}.

\begin{algorithm}[t]
\caption{Joint beamforming design for communication, detection and protection against adversarial detectors in RIS-aided ISAC} \label{algo_main}

\KwIn{  $\tau = 0$ , $\epsilon \!>\! 0$ }

Find an initial feasible point $\mathbf{W}^{(0)}$, $\boldsymbol{\Phi}^{(0)}$, $\mathbf{u}^{(0)}$ and $\mathcal{R}^{(0)}$ as described in Section~\ref{sec:init}\;

\Repeat{ $|| \mathbf{e}_{\mathsf{r}}^\mathsf{H} \boldsymbol{\Phi}^{(\tau)}\mathbf{G}_\mathsf{r} \mathbf{W}^{(\tau)}||_2^2 -  || \mathbf{e}_{\mathsf{r}}^\mathsf{H} \boldsymbol{\Phi}^{(\tau-1)}\mathbf{G}_\mathsf{r} \mathbf{W}^{(\tau-1)}\||_2^2 \leq \epsilon $}{
    
$ \tau \leftarrow \tau + 1$\;

Compute $\mathbf{W}^{(\tau)}$ and $\boldsymbol{\Phi}^{(\tau)}$ using~\textbf{Algorithm~\ref{algo1}}\;

Compute $\mathbf{u}^{(\tau)}$ and $\mathcal{R}^{(\tau)}$, and update $\mathcal{A}$, $\mathbf{G}_\text{r}$, $\mathbf{h}_{\text{r},k}$, $\mathbf{c}_{\text{r},\ell}$, $\mathbf{c}_{\text{a},\ell}$, $\mathbf{e}_{\text{r}}$, $\mathbf{e}_{\text{a}}$ and $\mathbf{\Phi}$  using~\textbf{Algorithm~\ref{algo2}}\;

}

\KwOut{ $\mathbf{W} = \mathbf{W}^{(\tau)}$, $\boldsymbol{\Phi} = \boldsymbol{\Phi}^{(\tau)}$, $\mathbf{u} = \mathbf{u}^{(\tau)}$, $\mathcal{R} = \mathcal{R}^{(\tau)}$}
\end{algorithm}

\subsection{Convergence}

\textit{1) Algorithm~\ref{algo1}}: The SCA method is guaranteed to converge to at least a local solution in the optimization variables, as demonstrated in~\cite{2014_Razaviyayn}. However, jointly optimizing $\mathbf{W}$ and $\boldsymbol{\Phi}$ (i.e., as presented in Algorithm~\ref{algo1}) can help to avoid local solutions that may arise when optimizing them sequentially~\cite{2022_Bezdek}.

\textit{2) Algorithm~\ref{algo2}}: The sub-optimization problem~\eqref{eq:SP2} shows that any optimal point produced by Algorithm~\ref{algo1} is also an optimal point of Algorithm~\ref{algo2} since the objective function is constant and all the constraints are satisfied. However, solving~\eqref{eq:SP3} allows allocating more RIS elements for reflection tasks without compromising the sensing SINR constraint. This extra degree of freedom, in turn, helps achieve better solutions for $\mathbf{W}$ and $\boldsymbol{\Phi}$ in the next iteration.

\subsection{Computational Complexity}
The computational complexities of Algorithm~\ref{algo1} and Algorithm~\ref{algo2} are dominated by solving the optimization problems~\eqref{eq:opt_JO} and~\eqref{eq:SP4}, respectively. The solutions to these optimization problems depend heavily on the iterative method used to solve them. To quantify the computational complexity associated with Algorithm~\ref{algo1} and Algorithm~\ref{algo2}, we assume that an interior-point method is employed to obtain the optimal solutions for both.

\textit{1) Algorithm~\ref{algo1}}: The complexity for solving SOCPs using the interior-point method can be upper bounded by $\mathcal{O}\big(\mathcal{M}^{0.5} (\mathcal{M} + \mathcal{B}) \mathcal{B}^2 \big)$, where $\mathcal{M}$ and $\mathcal{B}$ represent the number of constraints and variables, respectively~\cite{2004_Boyd}.  Based on this, we approximate the computational complexity of Algorithm~\ref{algo1} as $\mathcal{O} \big( I_1 (K^2 + L + N_\text{r})^{0.5} (K^2 + MK + L + N_\text{r}) (K^2 + MK + N_\text{r})^2 \big)$, where $I_1$ represents the maximum number of iterations needed for Algorithm~\ref{algo1} to converge. This computational complexity expression can be further upper-bounded by $\mathcal{O} \big( I_1 (K^2 + \max( L,MK) + N_\text{r})^{3.5}  \big)$.

\textit{2) Algorithm~\ref{algo2}}: It was also shown in~\cite{2004_Boyd} that the complexity for solving an optimization problem with an objective function comprising an $l_p$ norm with $\mathcal{B}$ constraints and $\mathcal{B}$ variables can be upper bounded by $\mathcal{O}\big(\mathcal{M}^{0.5} (\mathcal{M} + \mathcal{B}) \mathcal{B}^2 \big)$. Thus, we can approximate the computational complexity of Algorithm~\ref{algo2} as $\mathcal{O}\big(I_2 L^{0.5} (L + N_\text{a}) N_\text{a}^2 \big)$, where $I_2$ is the maximum number of iterations needed for Algorithm~\ref{algo2} to converge.

\textit{3) Algorithm~\ref{algo_main}}: Combining Algorithms~\ref{algo1} and Algorithms~\ref{algo2}, the per-iteration complexity of Algorithm~\ref{algo_main} is $\mathcal{O}\big( I_1 (K^2 + \max( L,MK) + N_\text{r})^{3.5} + I_2 L^{0.5} (L + N_\text{a}) N_\text{a}^2 \big)$. However, it can be noted that the complexity of Algorithm~\ref{algo2} is negligible compared to that of Algorithm~\ref{algo1}. Thus, we can approximate the per-iteration complexity of Algorithm~\ref{algo_main} as $\mathcal{O} \big( I_1 (K^2 + \max( L,MK) + N_\text{r})^{3.5}  \big)$.

\section{Numerical Results}   \label{sec:siml}
In this section, we present extensive simulation results to evaluate the performance of the proposed system. The following parameters are used unless stated otherwise. We assume that the BS is equipped with $M=4$ antennas, while the RIS comprises $N=64$ elements in total. The number of users is $K=4$, and the number of samples in the sensing area is set to $L=9$. The BS and the adversarial detector are assumed to have power budgets of $P_\text{max}=\unit[40]{dBm}$ and $\varrho^2=\unit[30]{dBm}$, respectively. The RCS areas are $\varsigma_\ell^2 = \unit[0.8]{m^2}$ for all $\ell \in \mathcal{L}$. The communication SINR thresholds for all users and the sensing SINR thresholds for all locations are set to $\Gamma_{\text{c},k} = \Gamma_{\text{s},\ell} = \unit[7.5]{dB}$ for all $k \in \mathcal{K}$ and for all $\ell \in \mathcal{L}$. Moreover, $\sigma_\text{s}^2 = \sigma_\text{d}^2 = \bar{\sigma}_{\text{c},k}^2 = \unit[-70]{dBm}$ for all $k \in \mathcal{K}$.

The BS and the RIS are located at $\unit[(0, 0, 0)]{m}$ and $\unit[(10, 0, 0)]{m}$, respectively. The channel between them is assumed to follow Rician fading with a Rician fading factor of \unit[3]{dB} and a path loss exponent of 2. The path gain at the reference distance of \unit[1]{m} is assumed to be \unit[-30]{dB}. The location of the $k$-th user is assumed to be uniformly distributed in a cuboidal region with opposite corners at $\unit[(7.5, 10, 0)]{m}$ and $\unit[(12.5, 15, 2)]{m}$. The channel between the RIS and the $k$-th user is also assumed to follow Rician fading with a Rician factor of \unit[3]{dB} and a path loss exponent of 2. The adversarial detector is located at a distance of \unit[8]{m}, with an azimuth angle of $60^\circ$ and an elevation angle $80^\circ$ relative to the reference element of the RIS. The sensing area is assumed to be at the same distance of \unit[8]{m} from the reference element of the RIS, extending over azimuth angles ranging from $45^\circ$ to $55^\circ$ and elevation angles ranging from $75^\circ$ to  $85^\circ$, with the $L$ points sampled uniformly in this area. Also, we average the performance metric over 250 independent channel realizations, with small-scale fading and user locations randomly varied.

We compare the performance of the proposed method with~\cite{2024_Magbool}, which aimed to maximize the sensing SINR while adhering to the other constraints, without protecting the sensing area. To enable a relevant comparison, we generalize~\cite{2024_Magbool} to sense $L$ locations and account for the presence of an adversarial detector's signal at the RIS. Since~\cite{2024_Magbool} uses a fixed RIS array configuration, we consider two configurations as benchmarks: sensing and communication with $\mathcal{R} = \mathcal{R}_1$, which we label as ``SC, $\mathcal{R} = \mathcal{R}_1$'', and sensing and communication with $\mathcal{R} = \mathcal{R}_2$, which we label as ``SC, $\mathcal{R} = \mathcal{R}_2$'', where $\mathcal{R}_1 = {1,\dots,40}$ and $\mathcal{R}_2 = {1,\dots,48}$. The purpose of including these benchmarks is to highlight the impact of neglecting the protection of the sensing area on the results.

We also include three versions of the proposed method. In the first version, we optimize $\mathbf{W}$, $\boldsymbol{\Phi}$, and $\mathbf{u}$ while setting $\mathcal{R} = \mathcal{R}_1$, which we label as ``Proposed, $\mathcal{R} = \mathcal{R}_1$.'' The second version is similar to the previous one but with $\mathcal{R} = \mathcal{R}_2$, which we label as ``Proposed, $\mathcal{R} = \mathcal{R}_2$.'' The final version applies Algorithm~\ref{algo_main} in full, which we label as "Proposed, adaptive $\mathcal{R}$."

\subsection{Convergence}

\begin{figure} 
         \centering
         \includegraphics[width=0.75\columnwidth]{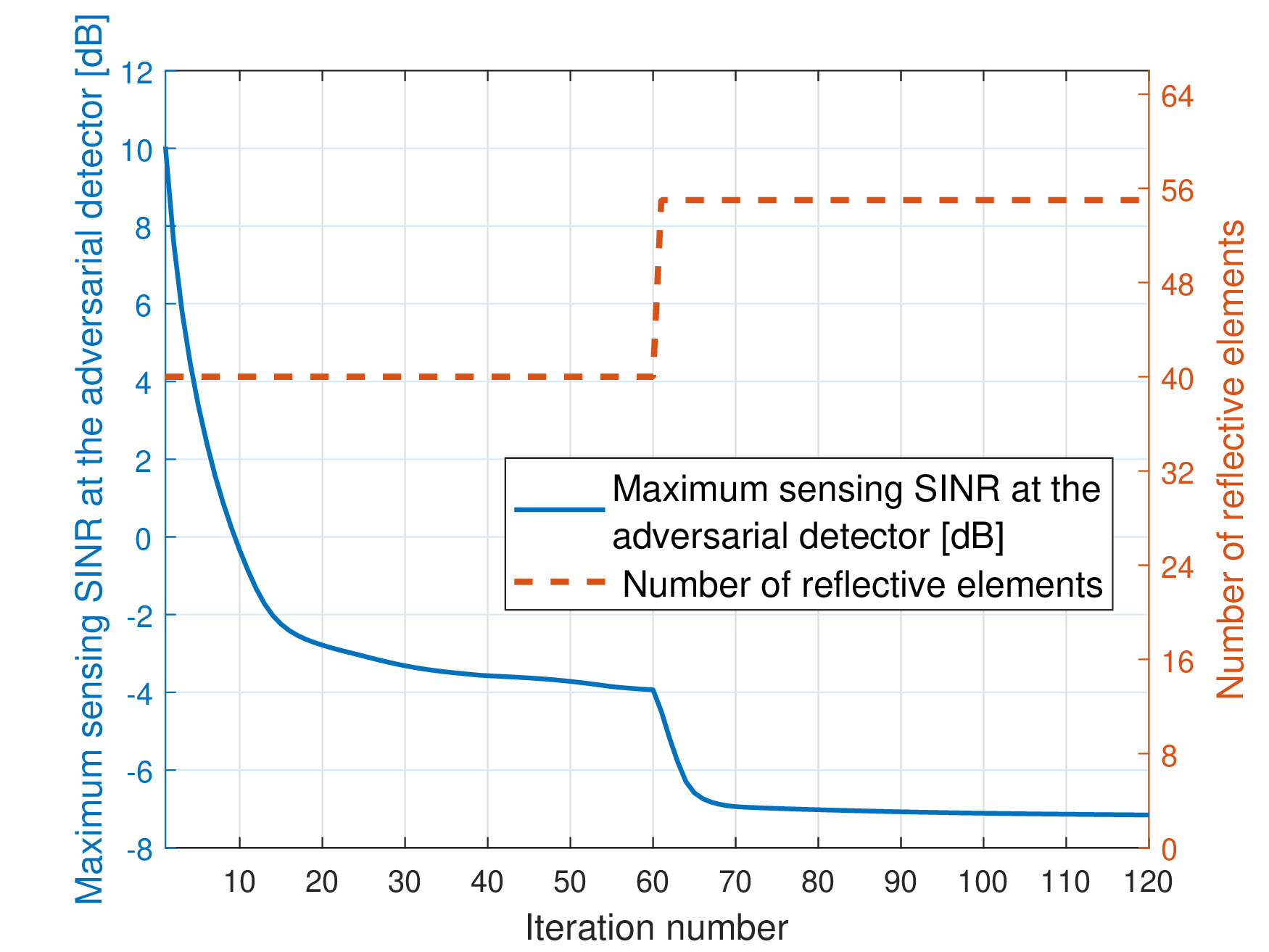}
         \vspace{0.05cm}
        \caption{Maximum sensing SINR at the adversarial detector (left y-axis) and number of reflecting elements (right y-axis) versus iteration number.}
        \label{fig:conv}
\end{figure}

\begin{figure*}
  \centering
  \begin{tabular}{c c c}
    \includegraphics[width=0.65\columnwidth]{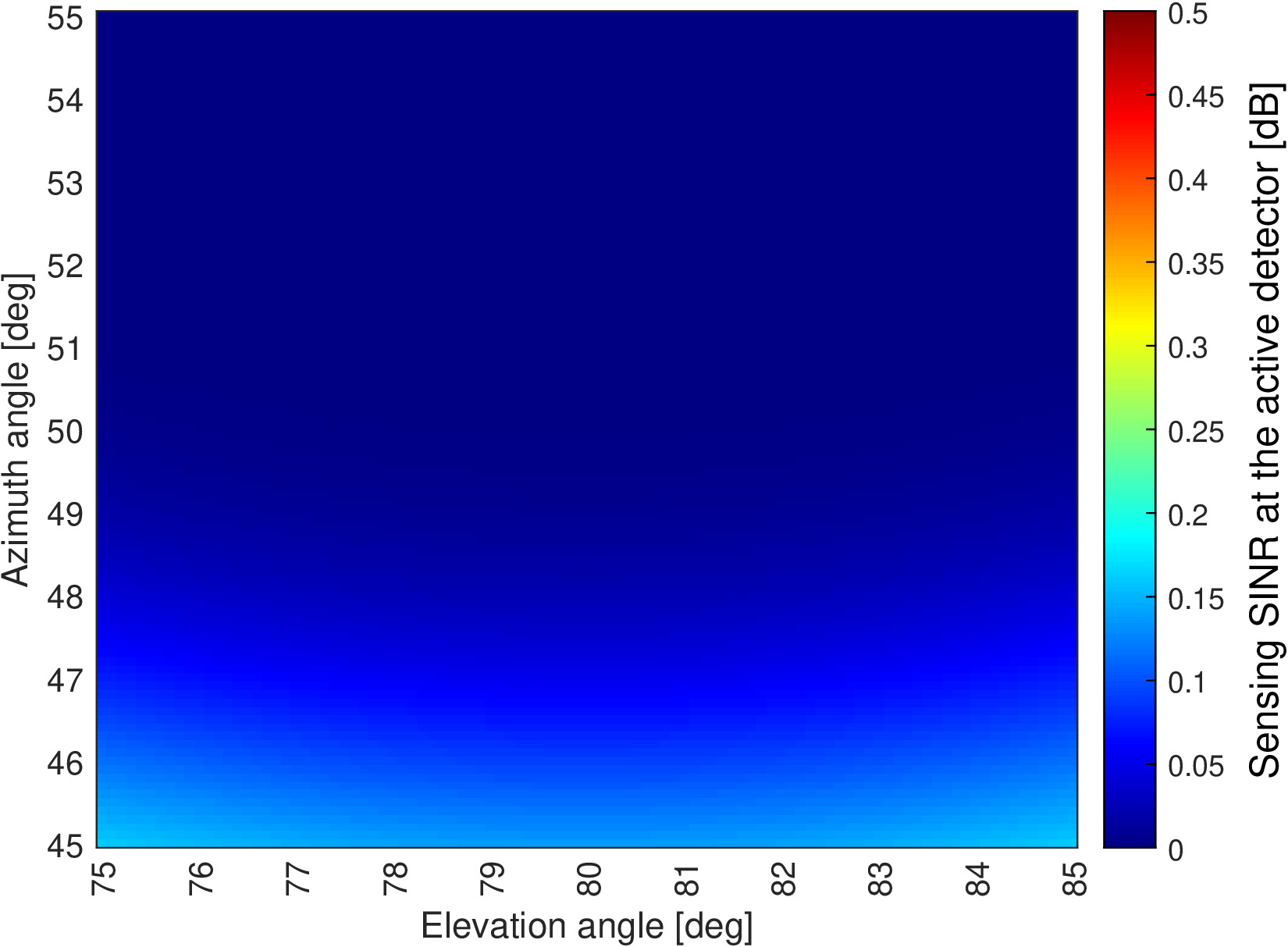} &
      \includegraphics[width=0.65\columnwidth]{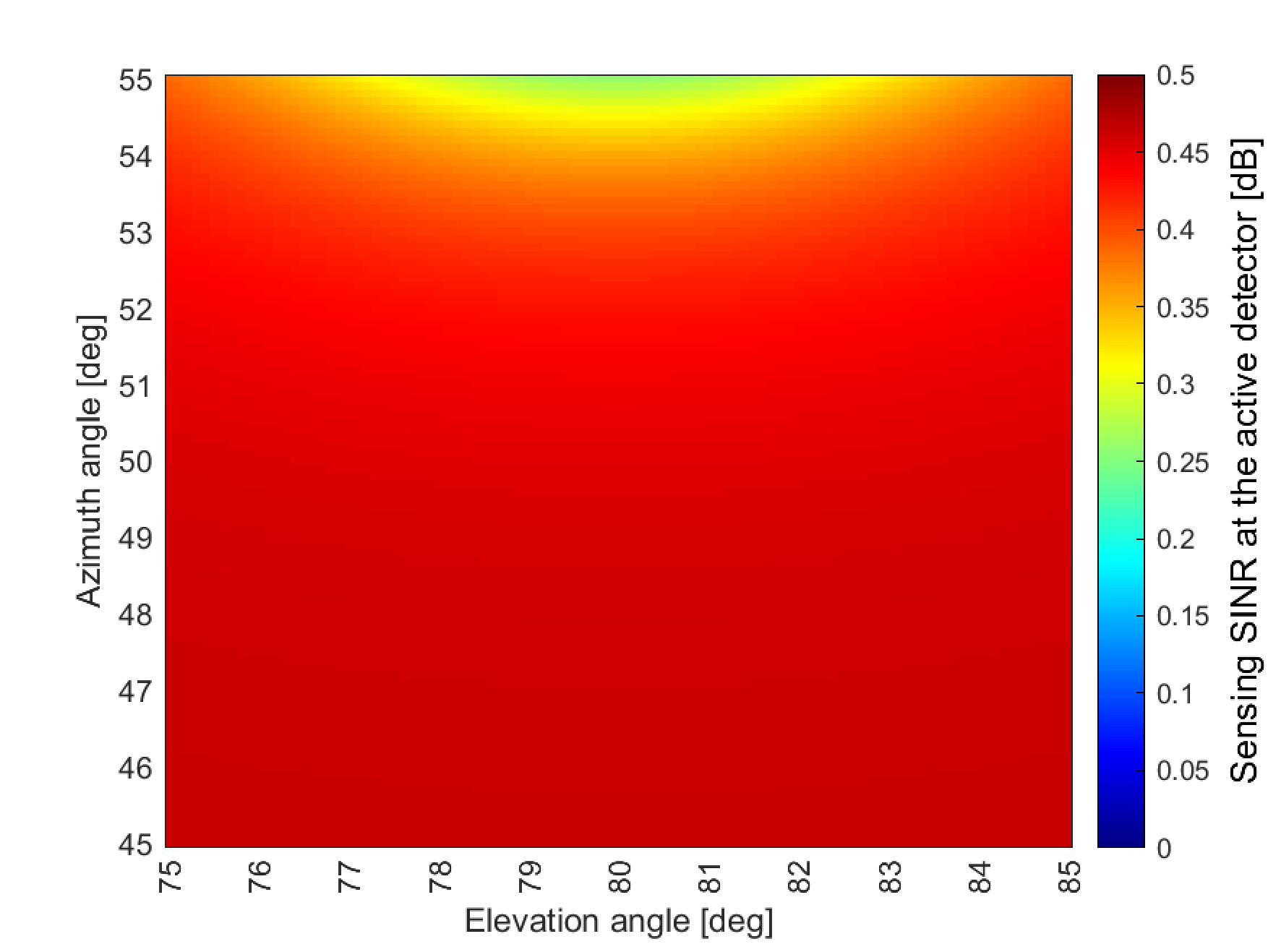} & \includegraphics[width=0.65\columnwidth]{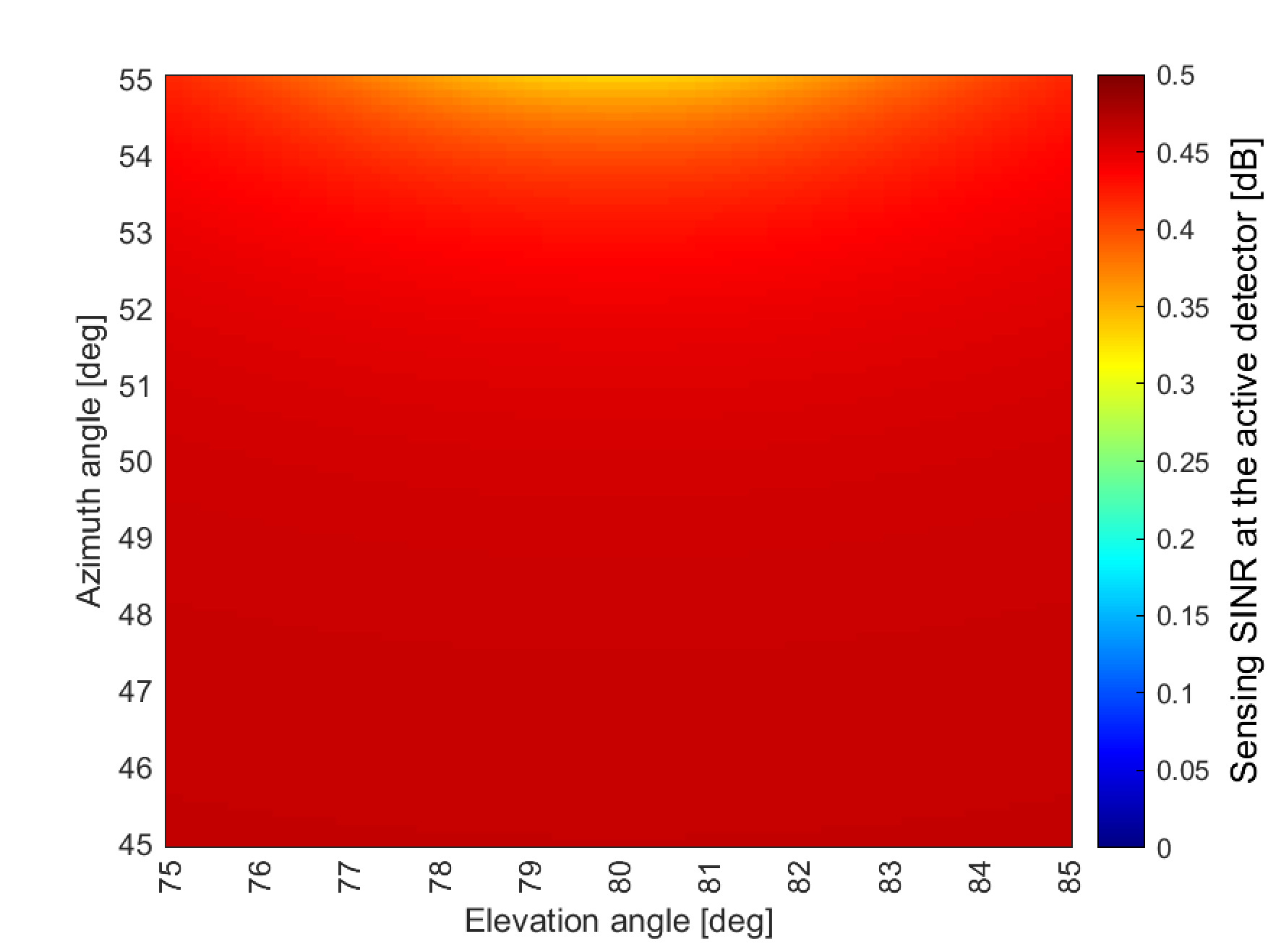} \\
      \scriptsize (a) SC, $\mathcal{R} = \mathcal{R}_2$.  &
      \scriptsize (b) Proposed, $\mathcal{R} = \mathcal{R}_2$. &
      \scriptsize (c) Proposed, adaptive $\mathcal{R}$. \\
  \end{tabular}
    \medskip
  \caption{Heatmaps of the FA probability at the adversarial detector within the sensing region for the baselines and the proposed method.}
  \label{fig:FA}
\end{figure*}

 \begin{figure*}
  \centering
  \begin{tabular}{c c c}
    \includegraphics[width=0.65\columnwidth]{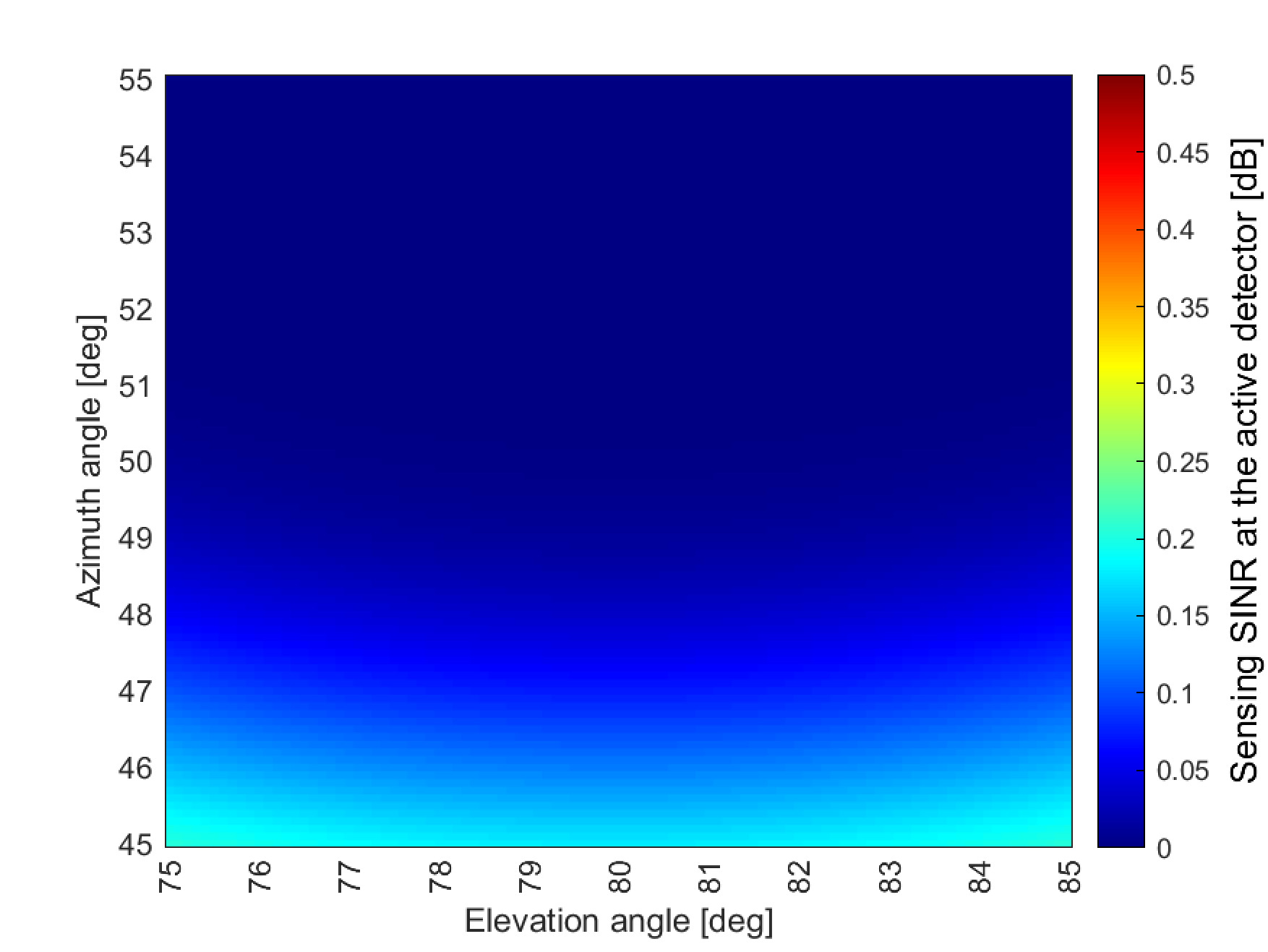} &
      \includegraphics[width=0.65\columnwidth]{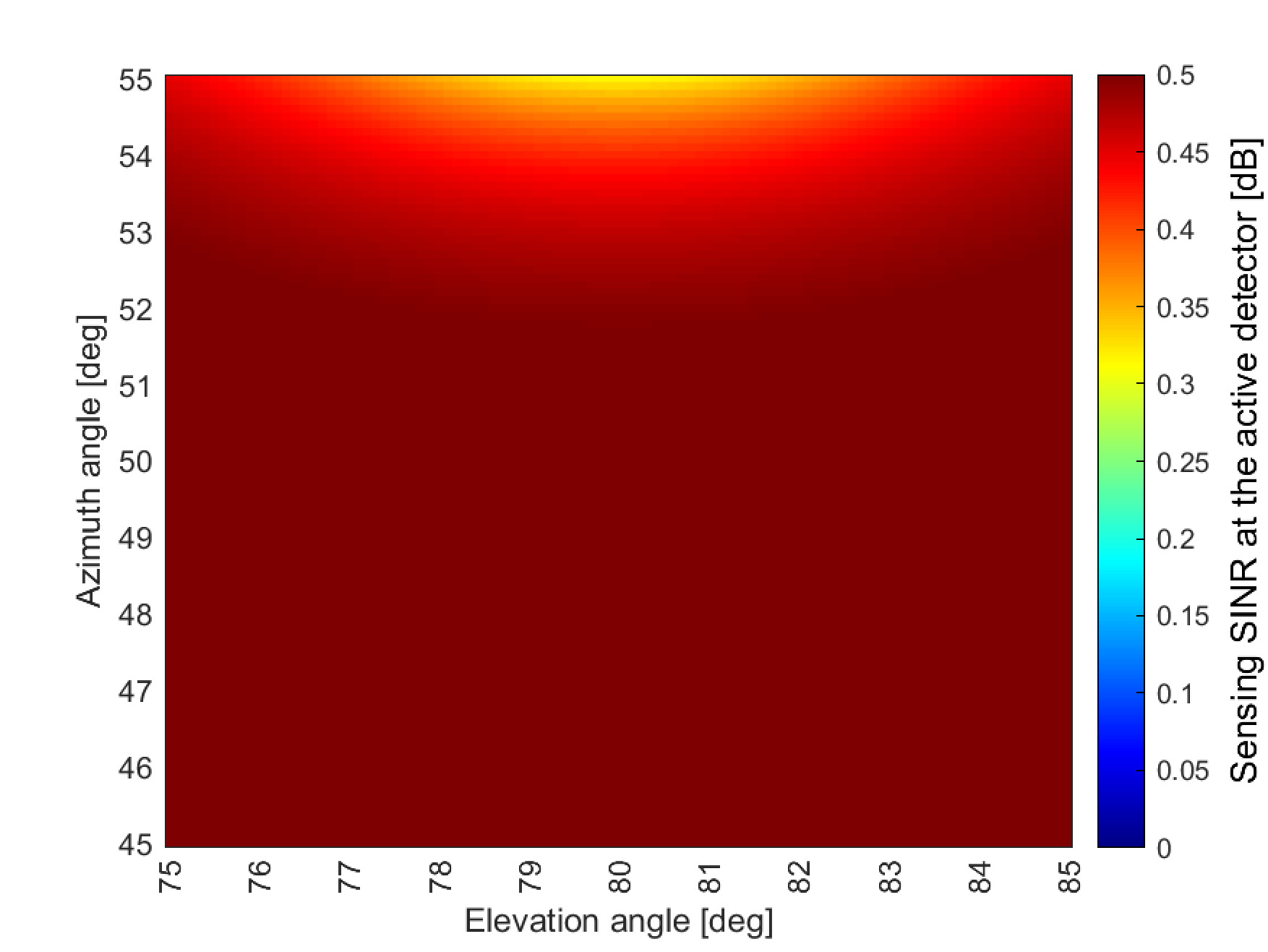} & \includegraphics[width=0.65\columnwidth]{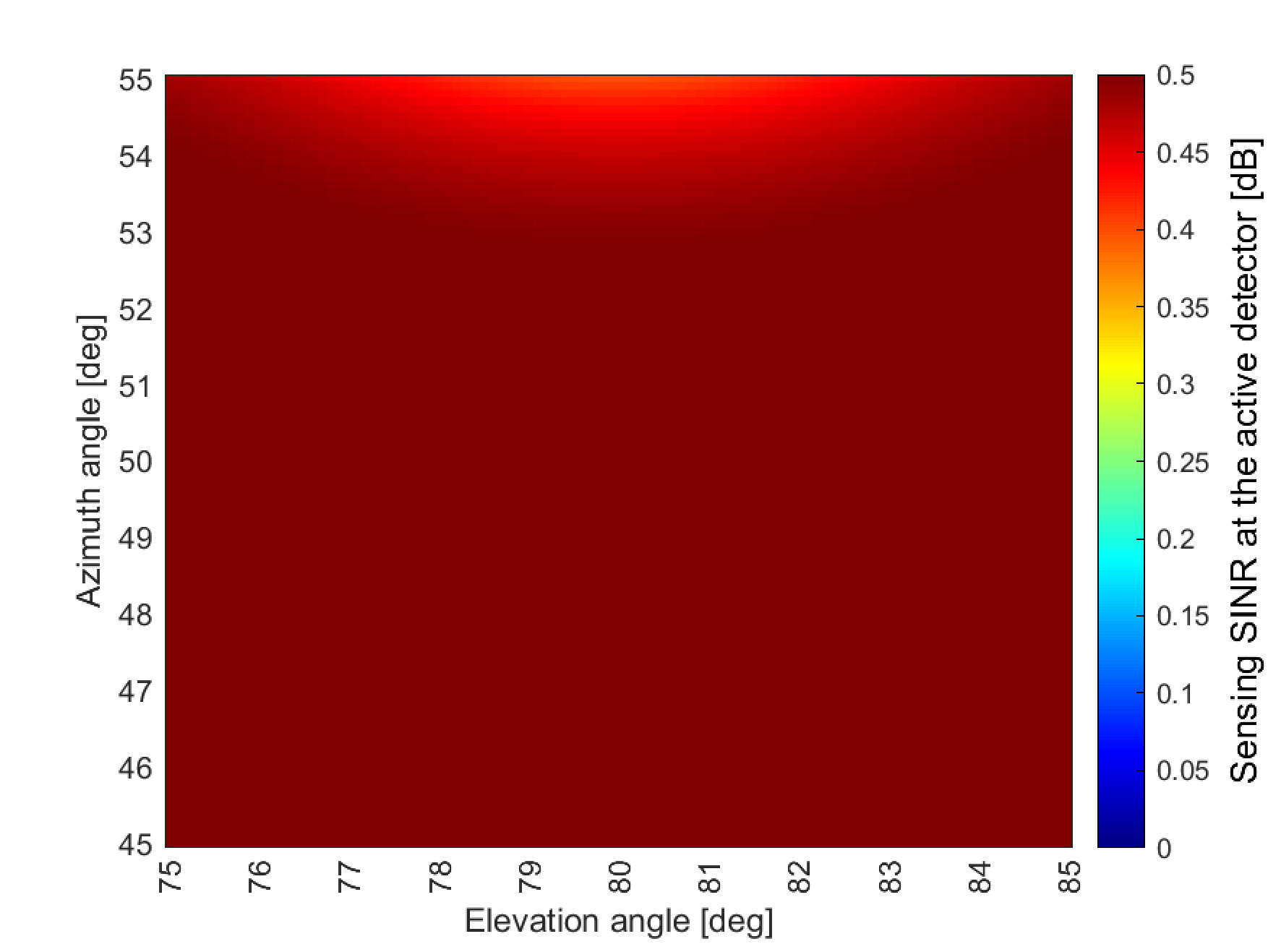} \\
 \scriptsize (a) SC, $\mathcal{R} = \mathcal{R}_2$.  &
      \scriptsize (b) Proposed, $\mathcal{R} = \mathcal{R}_2$. &
      \scriptsize (c) Proposed, adaptive $\mathcal{R}$. \\
  \end{tabular}
    \medskip
  \caption{Heatmaps of the MD probability at the adversarial detector within the sensing region for the baselines and the proposed method.}
  \label{fig:MD}
\end{figure*}

Fig.~\ref{fig:conv} shows the maximum sensing SINR at the adversarial detector among all locations at each iteration of the proposed algorithm. The number of reflecting elements is also plotted in the same figure using the right y-axis. It can be observed that the proposed algorithm successfully reduces the maximum sensing SINR from \unit[10]{dB} to \unit[-7]{dB}  upon re-application of Algorithm~\ref{algo1}. For the initial RIS element assignment, the maximum sensing SINR at the adversarial detector converges to \unit[-4]{dB} after applying Algorithm~\ref{algo1}. Subsequently, Algorithm~\ref{algo2} increases the number of reflecting elements from 40 to 55, further reducing the sensing SINR to \unit[-7]{dB}. Notably, Fig.~\ref{fig:conv} shows only two iterations of Algorithm~\ref{algo_main}, which is generally sufficient for convergence; in most cases, the number of reflecting elements remains unchanged after the first iteration of Algorithm~\ref{algo_main}.

\subsection{FA and MD Probabilities at the Adversarial Detector}

Fig.~\ref{fig:FA} and Fig.~\ref{fig:MD} illustrate heatmaps of the FA and MD probabilities at the adversarial detector at convergence within the sensing region for the baseline approaches and the proposed method, respectively (due to space limitations, we omit the figures corresponding to $\mathcal{R} = \mathcal{R}_1$). We assume here that the adversarial detector averages the energy of ten received echo signals to improve its detection performance (i.e., $T_\text{s} = 10$).

 {We observe from Proposition~\ref{col:1} that when $p_\ell = e^{-1}$ and $q_\ell = e$ (i.e., worst-case FA and MD probabilities for a single sample), averaging an infinitely large number of samples results in both $\bar{p}_\ell$ and $ \bar{q}_\ell$ converging to 0.5. This intuitively makes sense, as the received signal is dominated by noise and interference, and there is no way for the receiver to decide whether there is a target or not.}

\begin{figure} 
         \centering
         \includegraphics[width=0.75\columnwidth]{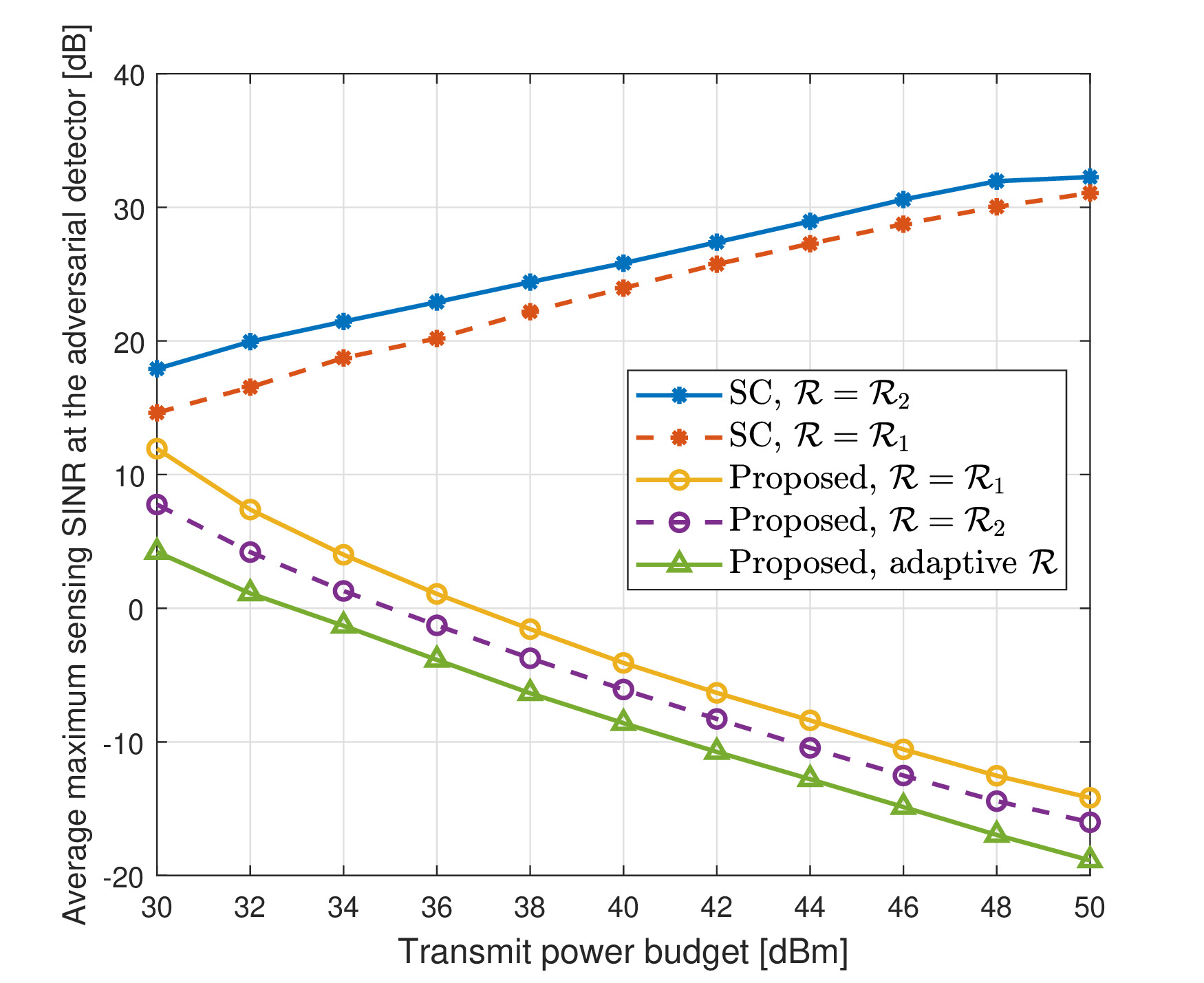}
         \vspace{0.05cm}
        \caption{Average maximum sensing SINR at the adversarial detector as the transmit power budget increases.}
        \label{fig:P_max}
\end{figure}

As shown in Fig.~\ref{fig:FA}~(a), when the protection of the sensing area is not a primary focus, the FA probability at the adversarial detector remains below 0.1 throughout the area, leaving the entire sensing area exposed to the adversarial detector. This situation can be significantly improved by minimizing the maximum SINR at the adversarial detector, as shown in Fig.~\ref{fig:FA}~(b); however, it can be seen that locations close to the adversarial detector have FA probabilities around 0.25. A further improvement is observed when the optimal division between reflecting and absorptive elements within the sensing area is carefully optimized, as depicted in Fig.~\ref{fig:FA}~(c), especially in locations close to the adversarial detector. The same argument can be applied to the MD probability in Fig.~\ref{fig:MD}. Ultimately, the proposed approach results in better FA and MD probabilities throughout the sensing region, compared to the baseline approaches.

\vspace{-0.3cm}
\subsection{Effect of Varying the Transmit Power Budget}

Fig.~\ref{fig:P_max} depicts the average maximum sensing SINR at the adversarial detector as the transmit power budget increases. Notably, when the sensing SINR is maximized without protection against the adversarial detector, the maximum sensing SINR at the adversarial detector increases. This occurs because more power is transmitted toward the sensing area by the RIS to detect the target. However, this increased power also benefits the adversarial detector, making it easier for it to detect the target's presence.

On the other hand, when protection is considered, the system becomes increasingly effective at shielding the sensing area from the adversarial detector as the transmit power budget rises. For instance, a reduction of approximately \unit[15]{dB} in the maximum sensing SINR is observed when the transmit power budget increases from \unit[30]{dBm} to \unit[40]{dBm}. Additionally, optimizing the assignments of the RIS elements can further decrease the sensing SINR at the adversarial detector by \unit[3]{dB} to \unit[5]{dB}.

\subsection{Effect of Varying the SINR Requirements}

\begin{figure} 
         \centering
         \includegraphics[width=0.75\columnwidth]{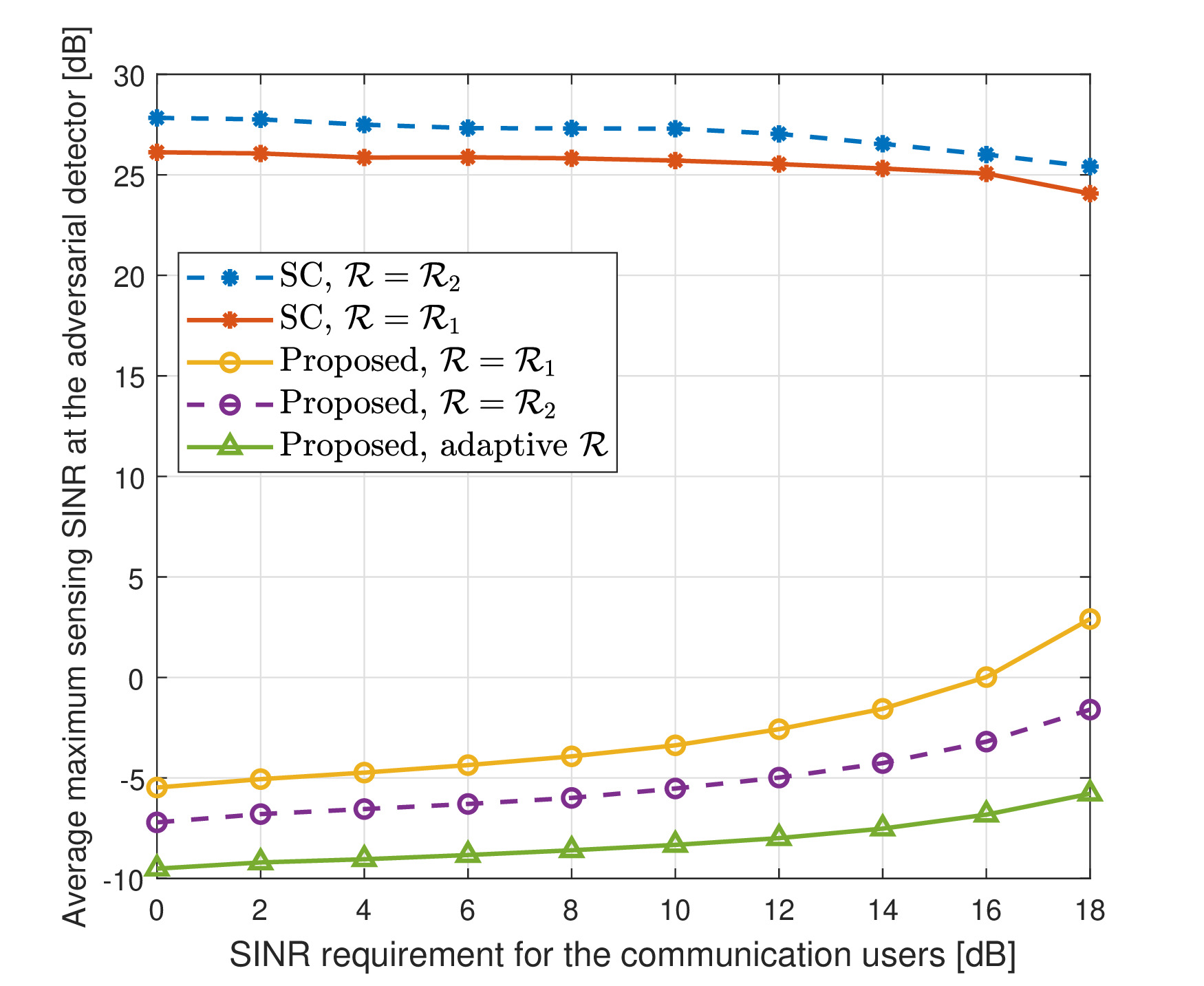}
         \vspace{0.05cm}
        \caption{Average maximum sensing SINR at the adversarial detector as a function of the SINR requirement for the communication users.}
        \label{fig:Gamma_c}
\end{figure}

\begin{figure} 
         \centering
         \includegraphics[width=0.75\columnwidth]{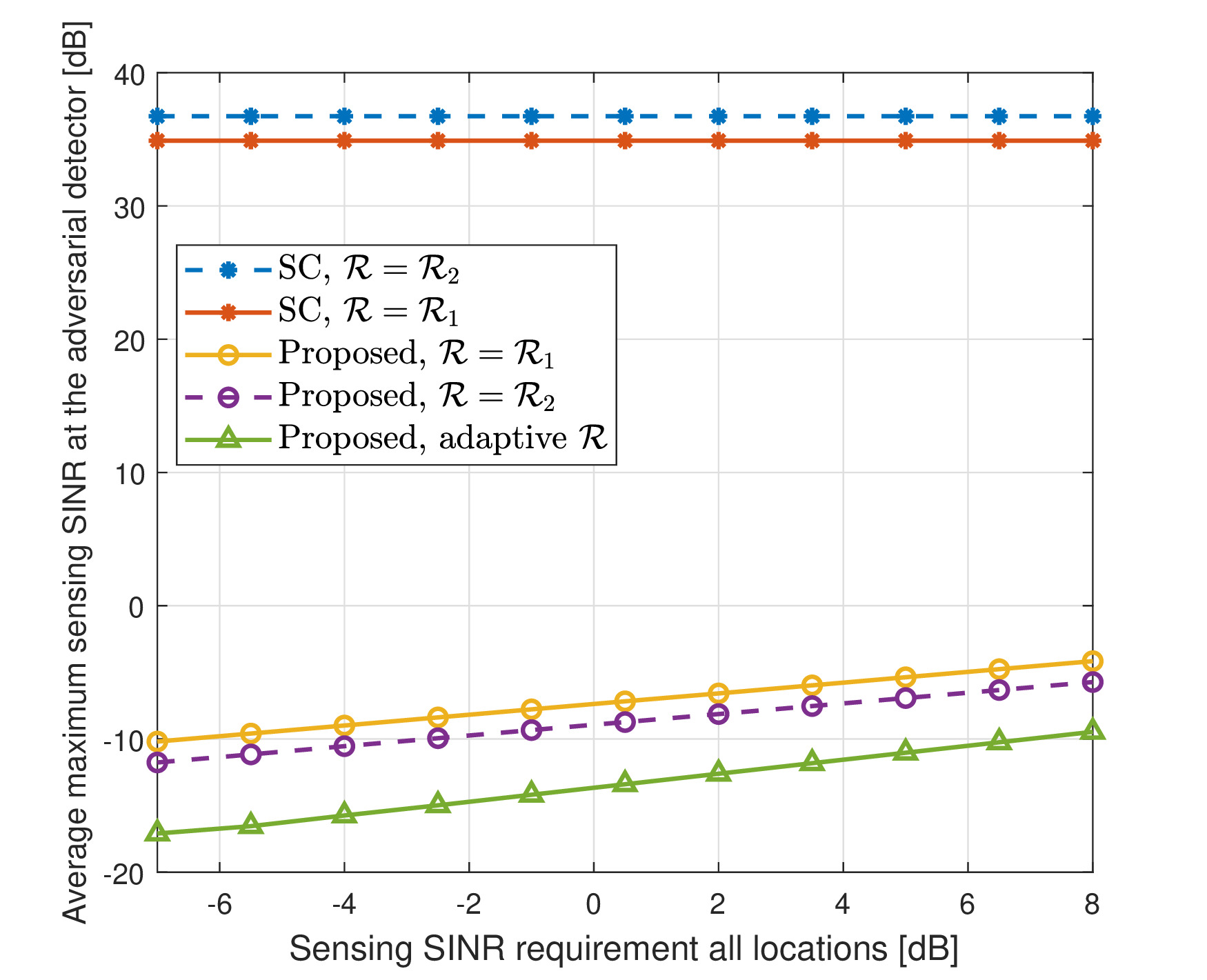}
         \vspace{0.05cm}
        \caption{Average maximum sensing SINR at the adversarial detector as a function of the sensing SINR requirement for all locations.}
        \label{fig:Gamma_s}
\end{figure}

Fig.~\ref{fig:Gamma_c} shows the average maximum sensing SINR at the adversarial detector as the SINR requirement for communication users varies. A higher communication SINR requirement can help reduce the sensing SINR at the adversarial detector, even without explicitly targeting protection. This occurs because more power is directed toward communication users. However, the sensing SINR remains relatively high in this scenario. 

Comparing the proposed method with fixed and adaptive RIS element functionality (fixed and adaptive set $\mathcal{R}$), we observe that thethe adaptive assignment is more effective in reducing the sensing SINR at the adversarial detector compared to a system with fixed RIS element functionality, especially at high values of the communication SINR requirement. This is because when a higher communication SINR is required, more elements are assigned as reflecting. This demonstrates that the proposed adaptive RIS element assignment is better suited to adapting to the system’s changing requirements. 

This capability is also evident in Fig.~\ref{fig:Gamma_s}, which illustrates the average maximum sensing SINR at the adversarial detector as the sensing SINR requirement varies across all locations. The figure shows that the proposed adaptive RIS setup provides better protection at low sensing SINR requirements, as more elements can function as reflecting elements, thereby generating greater interference for the adversarial detector.

\subsection{Effect of Varying the Number of RIS Elements}

 {
Fig.~\ref{fig:N_RIS} shows the average maximum sensing SINR at the adversarial detector as the number of RIS elements varies. To draw fair comparisons and to ensure that not too many elements are used for sensing, we replace the sets $\mathcal{R}_1$ and $\mathcal{R}_2$ with $\mathcal{R} =\{ 1,\dots ,0.625N\}$ and $\mathcal{R} = \{ 1,\dots ,0.75N\}$, respectively. This ensures that the ratio of reflecting elements to absorptive ones remains unchanged for different numbers of RIS elements.}

 {
When the sensing SINR is maximized, more RIS elements can direct a stronger signal toward the sensing area, which in turn significantly affects its protection. On the other hand, when the protection aspect is optimized, the sensing SINR at the adversarial detector is reduced. The benefit of adaptive RIS element assignment is evident in this case, especially with a larger number of RIS elements, as most elements can be assigned to serve reflection tasks. For instance, the adaptive assignment of RIS elements offers a \unit[2]{dB} improvement compared to the fixed assignment with $\mathcal{R} = { 1,\dots ,0.75N}$ at $N=64$. This improvement becomes even more pronounced at $N=512$, with a difference of \unit[6]{dB}.}

\begin{figure} 
         \centering
         \includegraphics[width=0.75\columnwidth]{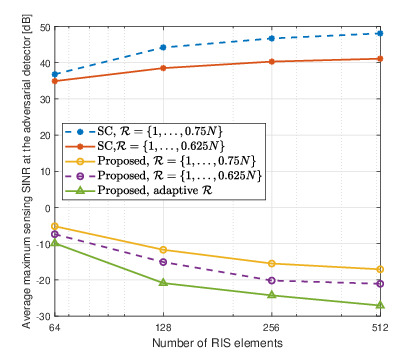}
         \vspace{0.05cm}
        \caption{Average maximum sensing SINR at the adversarial detector as the number of RIS elements increases.}
        \label{fig:N_RIS}
\end{figure}
\vspace{-0.3cm}
\section{Conclusion}  \label{sec:conc}
 This paper proposed an optimization framework for RIS-assisted ISAC while protecting the sensing region from an adversarial detector. Specifically, we minimized the maximum sensing SINR at the adversarial detector by jointly optimizing the transmit beamformer, RIS phase shifts, receive beamformer, as well as the division of RIS elements between reflection and sensing. The proposed AO framework, combined with a SCA approach, tackled this challenging problem while ensuring minimum sensing and communication SINR requirements. Simulation results demonstrate that the proposed method achieves a significant reduction in maximum sensing SINR at the adversarial detector, with optimal RIS element allocation providing an additional improvement over fixed RIS configurations.



%





\ifCLASSOPTIONcaptionsoff
  \newpage
\fi





\appendices
\vspace{-0.3cm}
\section{Proof of Proposition~\ref{prop1}} \label{AppndC}
First, we note that the sensing SINR  can be also expressed in terms of $\omega_{0}$ and $\omega_{\ell,1}$ as
\begin{equation}
     \gamma_{\text{s},\ell} (\omega_{0}, \omega_{\ell,1}) =   \frac{\omega_{\ell,1} - \omega_0}{ \omega_{0}} = \frac{\omega_{\ell,1} }{ \omega_{0}} - 1.
     \label{eq:gamma_s_omega}
\end{equation}

We can find the FA and MD probabilities for detecting the presence of a target at the $\ell$-th location as:
\begin{equation}
\begin{split}
         p_\ell   & = \text{Pr} ( |\mathbf{u}^\text{H} \mathbf{y}_{s,\ell} (t)|^2 \geq \bar{\omega}_{\ell} | \mathcal{H}_0) = 1 - \text{Pr} ( |\mathbf{u}^\text{H} \mathbf{y}_{s,\ell} (t)|^2  \\
        &  \leq \bar{\omega}_{\ell} | \mathcal{H}_0) = 1 - \int_0^{\bar{\omega}_{\ell}} \frac{1}{\omega_0} \text{exp} (- \frac{z}{\omega_0}) dz  \\
        & =\exp \Big(- \frac{\bar{\omega}_{\ell}}{\omega_0} \Big) = \Big( \frac{\omega_{1,\ell} }{\omega_0} \Big)^{- \frac{\omega_{1,\ell} }{\omega_{1,\ell} -  \omega_0}},
        \end{split}
        \label{eq:FAP}
\end{equation}
and
\begin{equation}
\begin{split}
        q_\ell & = \text{Pr} (|\mathbf{u}^\text{H} \mathbf{y}_{s,\ell} (t)|^2 \leq \bar{\omega}_{\ell} | \mathcal{H}_{1,\ell})= 1 - \exp  \Big(- \frac{\bar{\omega}_{\ell}}{\omega_{1,\ell}}\Big) \\
        & = 1 - \Big( \frac{\omega_{1,\ell} }{\omega_0} \Big)^{- \frac{ \omega_0}{\omega_{1,\ell} -  \omega_0}},
        \end{split}
        \label{eq:MDP}
\end{equation}
respectively. 

Substituting~\eqref{eq:gamma_s_omega} into~\eqref{eq:FAP} and~\eqref{eq:MDP}, we can express the FA and MD probabilities in terms of the sensing SINR as shown in~\eqref{eq:MDP_SINR} and~\eqref{eq:FAP_SINR}, respectively.

\vspace{-0.4cm}

\section{Proof of Lemma~\ref{lemma:1}} \label{AppndA}
We first present the following useful inequalities and identities:
\begin{enumerate}
    \item The first-order Taylor approximation of the squared norm of an arbitrary complex vector $\mathbf{v}$ around $\mathbf{v}_0$ is:
\begin{equation}
    \| \mathbf{v} \|_2^2  \geq 2 \Re \{ \mathbf{v}_0^\textsc{H}  \mathbf{v} \} - \| \mathbf{v}_0 \|_2^2.
    \label{R1}
\end{equation}
\item The following identities hold for any arbitrary complex vectors $\mathbf{v}_1$ and $\mathbf{v}_2$:
\begin{subequations}
\begin{align}
     \Re \{ \mathbf{v}_1^\textsc{H} \mathbf{v}_2 \}   & = \frac{1}{4} \Big( \| \mathbf{v}_1 + \mathbf{v}_2\|_2^2 - \| \mathbf{v}_1 - \mathbf{v}_2\|_2^2 \Big), \label{R2} \\
     \Im \{ \mathbf{v}_1^\textsc{H} \mathbf{v}_2 \}   & = \frac{1}{4} \Big( \| \mathbf{v}_1 -j \mathbf{v}_2\|_2^2 - \| \mathbf{v}_1 +j \mathbf{v}_2\|_2^2 \Big). \label{R3} 
\end{align}
\end{subequations}
\item The vectorization of a product of two matrices $\mathbf{V}_1 \in \mathbb{C}^{a \times b}$ and $\mathbf{V}_2 \in \mathbb{C}^{b \times c}$ is:
\begin{equation}
    \text{vec} (\mathbf{V}_1 \mathbf{V}_2) = ( \mathbf{V}_2^\text{T} \otimes \mathbf{I}_a) \text{vec} (\mathbf{V}_1 ).
\end{equation}
\end{enumerate}

We then re-express $\| \mathbf{a}^\textsc{H} \boldsymbol{\Phi} \mathbf{G}_\text{r} \mathbf{W}\|_2^2$ as
\begin{equation}
      \| \mathbf{a}^\textsc{H} \boldsymbol{\Phi} \mathbf{G}_\text{r} \mathbf{W}\|_2^2 = \| \boldsymbol{\phi}^\textsc{T} \mathbf{A}^* \mathbf{G}_\text{r} \mathbf{W}\|_2^2 = \| \mathbf{W}^\textsc{H} \mathbf{G}_\text{r}^\textsc{H} \mathbf{A} \boldsymbol{\phi}^*   \|_2^2,
\end{equation}
where $\boldsymbol{\phi} \triangleq \text{diag}(\boldsymbol{\Phi})$. Next, we utilize~\eqref{R1} to obtain the following lower bound:
  \begin{equation}
  \begin{split}
       \| \mathbf{W}^\textsc{H} \mathbf{G}_\text{r}^\textsc{H} \mathbf{A} \boldsymbol{\phi}^*   \|_2^2  & \geq 2 \Re \{ {\boldsymbol{\phi}_0^\textsc{T}} \mathbf{A}^* \mathbf{G}_\text{r}  {\mathbf{W}_0}   \mathbf{W}^\textsc{H} \mathbf{G}_\text{r}^\textsc{H} \mathbf{A} \boldsymbol{\phi}^*\}  \\
      & -\|   \underbrace{{\mathbf{W}_0}^\textsc{H} \mathbf{G}_\text{r}^\textsc{H} \mathbf{A}  \boldsymbol{\phi}_0^* }_{ \boldsymbol{\Xi}_3 (\mathbf{a}, \boldsymbol{\psi}_0)
     \boldsymbol{\psi}_0 } \|_2^2 
        \end{split}
        \label{eq:A1}
  \end{equation} 
 
 We use~\eqref{R2} to separate the optimization variables in the first term on the right-hand side of~\eqref{eq:A1} as
   \begin{equation}
  \begin{split}
       &  2 \Re \{ {\boldsymbol{\phi}_0^\textsc{T}} \mathbf{A}^* \mathbf{G}_\text{r}  {\mathbf{W}_0}   \mathbf{W}^\textsc{H} \mathbf{G}_\text{r}^\textsc{H} \mathbf{A} \boldsymbol{\phi}^*\}  = \frac{1}{2}  \| \mathbf{W} {\mathbf{W}_0^\textsc{H}}  \mathbf{G}_\text{r}^\textsc{H} \mathbf{A} {\boldsymbol{\phi}_0^*} \\ & +  \mathbf{G}_\text{r}^\textsc{H} \mathbf{A} \boldsymbol{\phi}^*\|_2^2  - \frac{1}{2}  \| \mathbf{W} {\mathbf{W}_0^\textsc{H}}  \mathbf{G}_\text{r}^\textsc{H} \mathbf{A} {\boldsymbol{\phi}_0^*} -  \mathbf{G}_\text{r}^\textsc{H} \mathbf{A} \boldsymbol{\phi}^*\|_2^2   \\
       &  = \frac{1}{2} \| \mathbf{W} {\mathbf{W}_0^\textsc{H}}  \mathbf{G}_\text{r}^\textsc{H} \mathbf{A} {\boldsymbol{\phi}_0^*} +  \mathbf{G}_\text{r}^\textsc{H} \mathbf{A} \boldsymbol{\phi}^*\|_2^2 \\
       & - \frac{1}{2} \| \underbrace{(\boldsymbol{\phi}_0^\textsc{H} \mathbf{A}^\textsc{T} \mathbf{G}_\text{r}^* \mathbf{W}_0^* \otimes \mathbf{I}_M) \text{vec} (\mathbf{W}) -  \mathbf{G}_\text{r}^\textsc{H} \mathbf{A} \boldsymbol{\phi}^*}_{\boldsymbol{\Xi}_1 (\mathbf{a}, \boldsymbol{\psi}_0)
     \boldsymbol{\psi}} \|_2^2.
        \end{split}
        \label{eq:A2}
  \end{equation}
  
  We use~\eqref{R1} to obtain a linear lower bound with respect to $\boldsymbol{\phi}$ and $\mathbf{W}$ in the first term on the right-hand side of~\eqref{eq:A2} as
  \begin{equation}
  \begin{split}
       & \frac{1}{2} \| \mathbf{W} {\mathbf{W}_0^\textsc{H}}  \mathbf{G}_\text{r}^\textsc{H} \mathbf{A} {\boldsymbol{\phi}_0^*} +  \mathbf{G}_\text{r}^\textsc{H} \mathbf{A} \boldsymbol{\phi}^*\|_2^2  \geq \Re \bigl\{ ( \mathbf{W}_0 {\mathbf{W}_0^\textsc{H}} \mathbf{G}_\text{r}^\textsc{H} \mathbf{A} \boldsymbol{\phi}_0^* \\
       & + \mathbf{G}_\text{r}^\textsc{H} \mathbf{A} {\boldsymbol{\phi}_0^*}  )^\textsc{H}  \big(  \mathbf{W} {\mathbf{W}_0} ^\textsc{H} \mathbf{G}_\text{r}^\textsc{H} \mathbf{A}^* {\boldsymbol{\phi}_0}^* +  \mathbf{G}_\text{r}^\textsc{H} \mathbf{A} \boldsymbol{\phi}^* \big) \bigl\} \\
       & - \frac{1}{2} \| \mathbf{W}_0 {\mathbf{W}_0^\textsc{H}} \mathbf{G}_\text{r}^\textsc{H} \mathbf{A} {\boldsymbol{\phi}_0^*} +  \mathbf{G}_\text{r}^\textsc{H} \mathbf{A} {\boldsymbol{\phi}_0^*} \|_2^2 \\
       & = \Re \bigl\{ \boldsymbol{\psi}_0^\textsc{H} \boldsymbol{\Xi}_2^\textsc{H} (\mathbf{a}, \boldsymbol{\psi}_0) \boldsymbol{\Xi}_2 (\mathbf{a}, \boldsymbol{\psi}_0) \boldsymbol{\psi} \bigl\}  - \frac{1}{2} \|\boldsymbol{\Xi}_2 (\mathbf{a}, \boldsymbol{\psi}_0) \boldsymbol{\psi}_0 \|_2^2. 
         \end{split}
         \label{eq:A3}
  \end{equation}
  By combining~\eqref{eq:A1},~\eqref{eq:A2}, and~\eqref{eq:A3}, we obtain the lower bound in~\eqref{eq:Lemm1}.

\vspace{-0.4cm}
 \section{Proof of Lemma~\ref{lemma:3}} \label{AppndB}
 \vspace{-0.2cm}
 We start by finding an upper bound on the square of the magnitudes of the real and imaginary parts of $  \mathbf{a}^\text{H} \boldsymbol{\Phi} \mathbf{G}_\text{r} \mathbf{w}_k $ as
\begin{subequations}
\begin{align}
     \big( \Re \{ \boldsymbol{\phi}^\text{T} \mathbf{A}^* \mathbf{G}_\text{r} \mathbf{w}_k \})^2 \leq \rho_k^2 \label{L1}, \\
      \big( \Im \{ \boldsymbol{\phi}^\text{T} \mathbf{A}^* \mathbf{G}_\text{r} \mathbf{w}_k \})^2 \leq \zeta_k^2. \label{L2}
\end{align}
\end{subequations}
 This indicates that
 \begin{equation}
     | \mathbf{a}^\text{H} \boldsymbol{\Phi} \mathbf{G}_\text{r} \mathbf{w}_k |^2 \leq \rho_k^2 + \zeta_k^2.
 \end{equation}
 Then from~\eqref{L1}, $\rho_k$ should be greater than or equal to both $\Re \{\mathbf{a}^\text{H} \boldsymbol{\Phi} \mathbf{G}_\text{r} \mathbf{w}_k \}$ and $-\Re \{ \mathbf{a}^\text{H} \boldsymbol{\Phi} \mathbf{G}_\text{r} \mathbf{w}_k \}$. Similarly, from~\eqref{L2}, $\kappa$ should be greater than or equal to both $\Im \{\mathbf{a}^\text{H} \boldsymbol{\Phi} \mathbf{G}_\text{r} \mathbf{w}_k \}$ and $-\Im \{ \mathbf{a}^\text{H} \boldsymbol{\Phi} \mathbf{G}_\text{r} \mathbf{w}_k \}$. We then use~\eqref{R2} and~\eqref{R3} to obtain the following conditions on $\rho$ and $\kappa$:
\begin{subequations}
\label{eq:tau_omega} 
\begin{align}
            \rho_k & \geq \pm \frac{1}{4} \big( \| \boldsymbol{\phi}^* + \mathbf{A}^* \mathbf{G}_\text{r} \mathbf{w}_k \|_2^2 - \| \boldsymbol{\phi}^* - \mathbf{A}^* \mathbf{G}_\text{r} \mathbf{w}_k \|_2^2 \big), \label{eq:tau1} \\
            \zeta_k &  \geq \pm \frac{1}{4} \big( \| \boldsymbol{\phi}^* - j\mathbf{A}^* \mathbf{G}_\text{r} \mathbf{w}_k \|_2^2 - \|\boldsymbol{\phi}^* + j\mathbf{A}^* \mathbf{G}_\text{r} \mathbf{w}_k \|_2^2 \big), \label{eq:omega1}
\end{align} 
\end{subequations}%
The first term on right-hand side of~\eqref{eq:tau1} (with the positive sign) can be re-written as $\begin{bmatrix}
    \mathbf{A}^*\mathbf{G}_\text{r} \otimes \mathbf{J}_k^\textsc{T} &  \mathbf{I}_{N_\text{r}}
\end{bmatrix} \boldsymbol{\psi} = \mathbf{\Pi}_k (\mathbf{a}) \boldsymbol{\psi}$. On the other hand, the second term on right-hand side of~\eqref{eq:tau1} (with the positive sign) is still concave. To tackle this, we use~\eqref{R1} as follows:
\begin{equation}
\begin{split}
    & \frac{1}{4} \|  \boldsymbol{\phi}^* - \mathbf{A}^* \mathbf{G}_\text{r} \mathbf{w}_k \|_2^2  \geq \frac{1}{2} \Re \{ ( \boldsymbol{\phi}_0^* - \mathbf{A}^* \mathbf{G}_\text{r} \mathbf{w}_{0,k} )^\textsc{H} (\boldsymbol{\phi}^* \\
    &  - \mathbf{A}^* \mathbf{G}_\text{r} \mathbf{w}_{k}) \}  - \frac{1}{4} \| \boldsymbol{\phi}_0^* - \mathbf{A}^* \mathbf{G}_\text{r} \mathbf{w}_{0,k} \|_2^2 \\
    & = \frac{1}{2} \Re \{ \boldsymbol{\psi}_0^\textsc{H}  \boldsymbol{\Pi}^\textsc{H} (- \mathbf{a}) \boldsymbol{\Pi} (- \mathbf{a}) \boldsymbol{\psi} \} - \frac{1}{4} \| \boldsymbol{\Pi} (- \mathbf{a}) \boldsymbol{\psi}_0 \|_2^2.
    \end{split}
    \label{eq:B3}
\end{equation}
Then, by combining~\eqref{eq:tau1} and~\eqref{eq:B3}, we can obtain the expression~\eqref{eq:in11} with the positive sign. In a similar manner, we can use~\eqref{R1} to tackle the concave terms in~\eqref{eq:tau1} (with the negative sign) to obtain the expression with the negative sign in~\eqref{eq:in11}. Also, the expression~\eqref{eq:in21} can be obtained in a similar manner from~\eqref{eq:omega1}.


\balance 
\bibliographystyle{IEEEtran}
\bibliography{IEEEabrv,Bibliography}
%


\end{document}